\newif\ifappx
\appxtrue

 \documentclass[journal,onecolumn]{IEEEtran}

\makeatletter
\let\origIEEEPARstart\IEEEPARstart
\renewcommand{\IEEEPARstart}[3][1.1]{%
  \def\@IEEEPARstartDROPDEPTH{#1\baselineskip}%
  \origIEEEPARstart{#2}{#3}%
}

\usepackage{graphicx,amsmath,bm,lineno}
\usepackage{amssymb,mathtools, url,hyperref}
\usepackage[ruled]{algorithm2e}
\usepackage{subcaption}
\usepackage[font=scriptsize]{caption}
\usepackage{psfrag}
\usepackage{graphicx}
\usepackage{amsthm}
\usepackage{assoccnt}
\usepackage{amsmath,subcaption}
\usepackage[ruled]{algorithm2e}
\usepackage{soul}
\usepackage{cite}

\SetKwInput{KwInput}{Input}
\SetKwInput{KwOutput}{Output}

\newtheorem{mydef}{Definition}
\newtheorem{mycoro}{Corollary}

\newtheorem{theorem}{Theorem}
\newtheorem{proposition}{Proposition}
\newtheorem{myremark}{Remark}
 
\SetKwInput{KwInput}{Input}
\SetKwInput{KwOutput}{Output}

\begin{document}
\title{State-Space Network Topology Identification from Partial Observations}

\author{Mario Coutino, \emph{Student Member, IEEE}, Elvin Isufi, \emph{Member, IEEE}, Takanori Maehara \\and Geert Leus, \emph{Fellow, IEEE}\thanks{M. Coutino and G. Leus are with the faculty of Electrical Engineering, Mathematics and Computer Science, Delft University of Technology, Delft, The
Netherlands. E. Isufi is with the Department of Electrical and Systems Engineering, University of Pennsylvania, Philadelphia, USA. T. Maehara is with AIP RIKEN, Tokyo, Japan. This research is supported in part by the ASPIRE project (project 14926 within the STW OTP programme), financed by the Netherlands Organization for Scientific Research (NWO) and in part by NSF CCF 1717120, ARO W911NF1710438, ARL DCIST CRA W911NF-17-2-0181, ISTC-WAS and Intel DevCloud. M. Coutino is partially supported by CONACYT and AIP RIKEN. E-mails: \{m.a.coutinominguez; g.j.t.leus\}@tudelft.nl, eisufi@seas.upenn.edu, takanori.maehara@riken.jp}
}

\markboth{Coutino et al. State-Space Network Topology Identification from Partial Observations}%
{Shell \MakeLowercase{\textit{et al.}}: Bare Demo of IEEEtran.cls for IEEE Journals}

\maketitle

\begin{abstract}
In this work, we explore the state-space formulation of a network process to recover, from partial observations, the underlying network topology that drives its dynamics. To do so, we employ subspace techniques borrowed from system identification literature and extend them to the network topology identification problem. This approach provides a unified view of the traditional network control theory and signal processing on graphs. In addition, it provides theoretical guarantees for the recovery of the topological structure of a deterministic continuous-time linear dynamical system from input-output observations even though the input and state interaction networks might be different. The derived mathematical analysis is accompanied by an algorithm for identifying, from data, a network topology consistent with the dynamics of the system and conforms to the prior information about the underlying structure. The proposed algorithm relies on alternating projections and is provably convergent. Numerical results corroborate the theoretical findings and the applicability of the proposed algorithm.
\end{abstract}



\IEEEpeerreviewmaketitle

\def\Fc{\mathcal{F}_{{\rm c}}}
\def\Sb{\bm S}
\def\Ub{\bm U}
\def\Ubt{\bm U^{-1}}
\def\phib{\bm \phi}
\def\diag{{\rm diag}}
\def\Phib{\bm \Phi}
\def\Lambdab{\bm\Lambda}
\def\G{\mathcal{G}}
\def\V{\mathcal{V}}
\def\E{\mathcal{E}}
\def\Re{\mathbb{R}}
\def\xb{\bm x}
\def\Lb{\bm L}
\def\Wb{\bm W}
\def\xbh{\bm x_{{\rm f}}}
\def\Fnv{\mathcal{F}_{{\rm nv}}}
\def\Fcev{\mathcal{F}_{{\rm cev}}}
\def\yb{\bm y}

\def\x{{\bm{x}}}
\def\xk{{\bm{x}}(k)}
\def\xkk{{\bm{x}}(k+1)}
\def\y{{\bm{y}}}
\def\yk{{\bm{y}}(k)}
\def\A{{\bm{A}}}
\def\J{{\bm{J}}}
\def\B{{\bm{B}}}
\def\u{{\bm{u}}}
\def\uk{{\bm{u}}(k)}
\def\D{{\bm{D}}}
\def\C{{\bm{C}}}
\def\P{{\mathcal{P}}}
\def\G{{\mathcal{G}}}
\def\V{{\mathcal{V}}}
\def\L{\bm S}
\def\Ts{\mathcal{T}_s}
\def\Os{\mathcal{O}_s}
\def\E{{\mathcal{E}}}
\def\Rn{{\mathbb{R}^{N}}}
\def\Rl{{\mathbb{R}^{L}}}
\def\Rm{{\mathbb{R}^{M}}}
\def\Rr{{\mathbb{R}}}
\def\Cn{{\mathcal{C}_{n}}}
\def\On{{\mathcal{O}_{n}}}
\newcommand{\rnk}[1]{{\rm rank}(#1)}

\def\F{{\bm{F}}}
\def\Gm{{\bm{G}}}
\def\H{{\bm{H}}}
\def\U{{\bm{U}}}
\def\J{{\bm{J}}}

\def\Gx{\G_x}
\def\Ex{\E_x}
\def\Gu{\G_u}
\def\Eu{\E_u}
\def\Ss{\bm{S}_*}
\def\Lx{\L_x}
\def\Lu{\L_u}
\def\fx{f_x}
\def\fu{f_u}

\def\fxd{\tilde{f}_x}
\def\fud{\tilde{f}_u}
\def\fsd{\tilde{f}_*}

\def\Gammaa{\mathcal{O}_{\alpha}}
\def\Phia{\mathcal{T}_{\alpha}}
\def\Psia{\mathcal{N}_{\alpha}}
\def\Phig{\mathcal{T}_{\gamma}}
\def\Psig{\mathcal{T}_{\gamma}}
\def\nb{\bm n}
\def\wb{\bm w}
\def\vb{\bm v}
\def\al{\alpha}
\def\ybak{\bm{y}_{k,\alpha}}
\def\ubak{\bm{u}_{k,\alpha}}
\def\nbak{\bm{n}_{k,\alpha}}
\def\wbak{\bm{w}_{k,\alpha}}
\def\vbak{\bm{v}_{k,\alpha}}

\def\Ub{\bm U}
\def\Yb{\bm Y}
\def\Nb{\bm N}
\def\Xb{\bm X}
\def\Pu{\bm\Pi_{\Ub^T}^{\perp}}
\def\Puf{\bm\Pi_{\Ub_2^T}^{\perp}}
\def\Jo{\bm J}
\def\Gammag{\mathcal{O}_{\gamma}}
\def\be{\beta}
\def\ga{\gamma}

\def\Lambdabx{\bm\Lambdab_x}
\def\Lambdabu{\bm\Lambdab_u}

\def\lambdabx{\bm \lambda_x}
\def\Lambdabx{\bm\Lambda_x}

\def\Uaa{\bm U_{\ga,n}}

\def\Mxcal{\mathcal{M}_x}


\section{Introduction}

\IEEEPARstart{T}{he} topology of networks is fundamental to model the interactions between entities and to improve our understanding about the processes defined over them. We can find examples of such processes in transportation networks~\cite{deri2016new}, brain activity~\cite{sporns2010networks}, and epidemic dynamics or gene regulatory networks~\cite{mittler2004reactive}, to name a few. The coupling between the process and the network topology has led to the extension of signal processing (SP) techniques to tools that take into account the network structure to define signal estimators~\cite{narang2013signal,di2017adaptive,chen2015signalrecovery}, filters~\cite{shuman2013emerging,isufi2017autoregressive,minguez2019advances}, and optimal detectors~\cite{hu2016matched,chepuri2016subgraph,isufi2018Blind}.

While in several scenarios, the network structure is available, in many others, it is unknown and needs to be estimated. This not only for enhancing data processing tasks but also for data \emph{interpretability}, i.e., the network topology provides an abstraction for the underlying data dependencies. Therefore, retrieving the network structure or the dependencies of the involved members (variables) has become a research topic of large interest~\cite{kalofolias2016learn,mateos2018inference,segarra2017network,karanikolas2016multi,shen2017kernel,lu2017closed,zhou2007topology,shen2017tensor,shafipour2017network}.

Despite that many works focus on the problem of topology identification~\cite{mateos2019connecting,8347160} or Gaussian graphical modeling (GGM)~\cite{dempster1972covariance, lauritzen1996graphical}, most of these approaches only leverage a model based on so-called \emph{graph filters}~\cite{segarra2017optimal} or enforce a particular structure by a penalized likelihood approach with sparsity constraints~\cite{yuan2007model,kumar2019unified}. Among the works that consider an alternative interaction model e.g.,~\cite{shafipour2017network,shen2017kernel,ioannidis2018semi, ioannidis2019semi}, only few of them considers the network data as states of an underlying process. However, none of the above works study the case where the input, i.e., excitation or probing signal of the process and the process itself evolve according to different topologies.

In many instances, physical systems can be defined through a state-space formulation with known dynamics. An example is the diffusion of molecules in tissue. This process is used to analyze brain functions by mapping the interaction of the molecules with obstacles~\cite{iturria2008studying}. Besides the mapping of the brain, the area of neural dynamics presents the problem of network design in transport theory~\cite{hertz2018introduction}. In such applications, the topology that provides a stable desired response needs to be found, following a differential equation. Finally, we recall the problem of finding the connections between reactants in chemical reaction networks~\cite{angeli2009tutorial}. Here, the evolution of the number of the molecules in a solution is governed by the interaction of the reactants present in it. Hence, to understand the underlying chemical process, the relation between reactants is required. Considering these examples, it is clear that a more general approach, parting from first-principles, to find the underlying connections is required. So, in this work, we focus on the problem of \emph{retrieving the network structure} of a process modeled through a \emph{deterministic continuous-time dynamical system} whose system matrices depend on the underlying topology.


We further remark that all above works estimate the network topology under the assumption that the observations are available on all entities. We here depart from this assumption and consider the problem of network topology inference from partial observations. In this way, the focus is not only on retrieving a topology that describes the dynamics of the process under full network access, but also for cases where observations from a subset of the entities are only available. To address these tasks, we devise a framework that allows estimating the network topology from partial observations up to ambiguities defined through an equivalence class of restricted cospectral graphs~\cite{chung1997spectral}.


\subsection{Overview and main contributions}

Network topology inference from partial measurements is generally an ill-posed problem and has not been addressed by existing methods. In this work, we introduce a first-principles approach based on state-space models for model-driven topology estimation. Our contributions that broaden the state-of-the-art are threefold.

\begin{itemize}
    \item[--] We propose a first-order differential graph formulation for systems whose dynamics are described by deterministic continuous-time linear models. Under this model, it is possible to obtain a first-principles based network topology identification framework that leverages the geometric structure in its state-space description. Further, we provide conditions under which we can retrieve the network topology from sampled observations.

    \item[--] We extend subspace techniques from the system identification literature to network topology inference. These methods leverage the embedded geometrical structure of the state-space model that describes the dynamics. We provide ways to enforce the underlying dynamical structure on the topology identification process thereby guaranteeing system consistency, i.e., we enforce the structure present in the input-output data relations such that the estimated topology achieves the same dynamics as the original process.
    
    \item[--] We introduce the problem of network topology identification from partial observations. We mathematically analyze this problem and show it is ill-posed. We further describe the ambiguities present when recovering the network topology for measurements that do not uniquely identify the underlying structure. We develop an algorithm to estimate the network structure in the current setting. This algorithm relies on the alternating projections (AP) approach~\cite{bauschke1996projection} and is provably globally convergent. We prove that under mild conditions the AP method converges locally with a linear rate to a feasible solution. Finally, we extend the topology inference problem from partial observations to scenarios where incomplete or inaccurate process dynamics are present. For these cases, we provide a mathematical analysis and conditions that guarantee the convergence of the AP method in estimating a feasible network topology.
    
\end{itemize}


\subsection{Outline and notation}
This paper is organized as follows. Section~\ref{sec.problem} formulates the problem of network topology inference for continuous-time dynamical systems from sampled observations. Section~\ref{sec.model} introduces a first-order differential graph model and its state-space description. The system identification framework for the proposed graph-based model is introduced in Section~\ref{sec.state}. Section~\ref{sec.partial} analyzes the ambiguity of network topology identification from partial observations and provides an AP method to find a feasible network structure. Section~\ref{sec.consistency} discusses system consistency constraints that can be enforced into the AP method to match the network dynamics. Section~\ref{sec.numerical} corroborates the derived theory with numerical results. Finally, section~\ref{sec.conclusions} concludes the paper.

We adopt the following notation. Scalars, vectors, matrices and sets are denoted by lowercase letters $(x)$, lowercase boldface letters $(\bm x)$, uppercase boldface letters $(\bm X)$, and calligraphic letters $(\mathcal{X})$, respectively. $[\bm X]_{i,j}$ denotes the $(i,j)$th entry of the matrix $\bm X$ whereas $[\bm x]_i$ represents the $i$th entry of the vector $\bm x$. $\bm X^T$ and $\bm X^{-1}$ are the transpose and the inverse of $\bm X$, respectively. The Moore-Penrose pseudoinverse of $\bm X$ is denoted by $\bm X^{\dagger}$. ${\rm vec}(\cdot)$ is the vectorization operation. ${\rm bdiag}(\bm X,\bm Y)$ denotes a block diagonal matrix whose blocks are given by the matrices $\bm X$ and $\bm Y$. $\bm I$ is the identity matrix of appropriate size. $\Vert \bm X\Vert_{\rm F}$ and $\Vert\bm X\Vert_2$ denote the Frobenius- and $\ell_2$-norm of $\bm X$, respectively. ${\rm span}(\cdot)$ and ${\rm rank}(\cdot)$ are the span and rank of a matrix, respectively. Finally, we use $[K]$ to denote the set $\{1,2,\ldots,K\}$ and $\mathcal{D}_K$ to denote the set of $K\times K$ diagonal matrices.

\section{Problem Statement}\label{sec.problem}

Consider a set of $N$ nodes $\V=\{v_{1},\ldots,v_{N}\}$ representing cooperative agents such as sensors, individuals, and biological structures. On top of these agents, a process $\P$ is defined that describes the evolution through time of the agent signals $\x(t)\in\Re^{N}$. The signal $\x(t)$ is such that the $i$th component $x_i(t)$ represents the signal evolution of agent $v_i$. The agent interactions w.r.t. the evolution of the signal $\x(t)$ are captured by a graph $G_x = (\V,\E_x)$, where $\E_x$ is the edge set of this graph. We consider that process $\P$ is represented by a first-order differential model
\begin{subequations}\label{eq.modelInit}
	\begin{eqnarray}
		\partial_t\x(t) &=& h(\V,\E_x,\E_u,\x(t),\u(t))\\
		\bm y(t) &=& c(\V,\E_x,\E_u,\x(t),\u(t)),
	\end{eqnarray}
\end{subequations}
where $\partial_t\x(t):=d\x(t)/dt$ and $h(\cdot)$ and $c(\cdot)$ are maps that describe respectively the dynamics of the signal $\x(t)$ and the observables $\y(t)$. In model \eqref{eq.modelInit}, $\u(t)$ is the (known) system input and $\E_u$ is the edge set of another graph $G_u = (\V,\E_u)$ that captures the interactions between the elements of $\V$ for $\u(t)$. Put simply, process~\eqref{eq.modelInit} describes the evolution of the signal $\x(t)$ under the influence of the maps $h(\cdot)$ and $c(\cdot)$ and the network topologies $G_x = (\V,\E_x)$ and $G_u = (\V,\E_u)$. For future reference, we will represent both graphs as $G_* = \{\V,\E_*\}$, where $``*"$ is a space holder for $x$ and $u$.

We can differently define the process $\P$ through the set
\begin{equation}
    \P := \{h(\cdot),c(\cdot)\},
    \label{eq.pdescript}
\end{equation}
which contains the interactions in the system. While process $\P$ describes a continuous-time process, we usually have access to \emph{sampled realizations of it}, i.e., the observables $\y(t)$ are collected on a finite set of time instances $\mathcal{T}~:=~\{t_1,t_2,\ldots,t_T\}$. As a result, we might also want to consider \emph{discrete-time} approximations of~\eqref{eq.pdescript}, where we either have a discrete-time realization of $\P$ and/or a finite number of observables $\mathcal{Y} := \{\y(t)\}_{t\in\mathcal{T}}$.

With this in place, we ask the following question: \emph{how can we retrieve the network topologies $\G_x$ and $\G_u$ for the agent signal $\x(t)$ and input signal $\u(t)$ given the known process $\P$ [cf. \eqref{eq.pdescript}], input signal $\bm u(t)$, and observables in $\mathcal{Y}$?}

In this work, we answer the above question by employing results from Hankel matrices~\cite{takens1981detecting} and linear algebra whose foundations lie in the system identification theory~\cite{viberg1995subspace}. The approach we will consider is based on subspace techniques that are by definition \emph{cost function free}. This differs from the commonly used techniques in network topology identification~\cite{mateos2019connecting} where the learned topology heavily depends on the considered cost function (e.g., smoothness or sparsity). If this prior knowledge is incorrect, it might lead to structures not related to physical interactions. The adopted techniques provide theoretical insights in when and how the underlying network structures can be identified. Furthermore, they have two other benefits. First, they impose none parameterization on the dynamical model and, therefore, avoid solving nonlinear optimization problems as in prediction-error methods~\cite{ljung2002prediction}. Second, they allow identifying $G_* = \{\V,\E_*\}$ from sampled data, i.e., by having access only to a subset of the observables $\y(t)$.

We finally remark that although more general models can be considered, e.g., higher-order differential models, we restrict ourselves to first-order models to ease exposition and provide a thorough and stand-alone work. Nevertheless, first-order differential models are of broad interest as they include diffusion processes --see \cite{chung2007diffusion} and references therein.

\section{First Order Differential Graph Model}\label{sec.model}

The structure of $G_* = \{\V,\E_*\}$ is mathematically represented by a matrix $\Ss$ (sometimes referred to as a \emph{graph shift operator}~\cite{sandryhaila2013discrete,ortega2018graph}) that has as candidates the graph adjacency matrix, the graph Laplacian, or any other matrix that captures the connections (relations) of the elements in the network. We then consider that process $\P$ is described through the linear continuous-time dynamical system
\begin{subequations}\label{eq.lds}
	\begin{eqnarray}
		\partial_t\x(t) &=& \fx(\Lx)\x(t) + \fu(\Lu)\bm u(t)\in\mathbb{R}^{N}\\
		\bm y(t) &=& \bm C\x(t) + \D\bm u(t)\in\mathbb{R}^{L},
	\end{eqnarray}
\end{subequations}
where $\C\in\mathbb{R}^{L\times N}$ and $\D~\in~\mathbb{R}^{L\times N}$ are system matrices related to the observables $\bm y(t)$. The matrix function $f_* : \mathbb{R}^{N\times N} \rightarrow \mathbb{R}^{N\times N}$ is defined via the Cauchy integral~\cite{higham2008functions}
\begin{equation}\label{eq.intDef}
	f_*(\L_*) := \frac{1}{2\pi i}\int_{\Gamma_{f_*}}f_{\rm s,*}(z)R(z,\L_*)dz,
\end{equation}
where $f_{\rm s,*}(\cdot)$ is the scalar version of $f_*(\cdot)$ and is analytic on and over the contour $\Gamma_{f_*}$. Here, $R(z,\L_*)$ is the resolvent of $\L_*$ given by
\begin{equation}
	R(z,\L_*) := (\L_* - z\bm I)^{-1}.
\end{equation}
From the definition of the dynamical system in~\eqref{eq.modelInit}, it can be seen that \eqref{eq.lds} is a proper representation of the family of first-order (linear) differential models, where the system matrices are \emph{matrix functions} of the matrix representing the graph. Observe that although restricted to the class of linear models, \eqref{eq.lds} captures different settings of practical interest such as diffusion on networks~\cite{chung2007diffusion}, graph filtering operations~\cite{segarra2017optimal,minguez2019advances}, random walks~\cite{lovasz1993random}, and first-order autoregressive graph processes~\cite{isufi2018forecasting}.

The corresponding discrete-time state-space system related to~\eqref{eq.lds} is
\begin{subequations}\label{eq.ltia}
	\begin{eqnarray}
		\xkk &=& \fxd(\Lx)\xk + \fud(\Lu)\uk +\bm w(k),\\
		\yk &=& \C\xk + \D\uk + \bm v(k),
	\end{eqnarray}
\end{subequations}
where $\tilde{f}_*(\cdot)$ is a matrix function (to be specified in the sequel) and $\xk \in \Rr^{N}$, $\uk\in\Rr^{N}$, and $\yk\in\Rr^{L}$ are the discrete counterparts of $\x(t)$, $\u(t)$, and $\y(t)$, respectively. The variables, $\bm w(k)$ and $\bm v(k)$ represent perturbations in the states and additive noise in the observables, respectively. By defining then the matrices $A(\Lx) := \fxd(\Lx)$ and $\bm B({\Lu}) := \fud(\Lu)$, the connection between the continuous-time \eqref{eq.lds} and the discrete-time representation \eqref{eq.ltia} is given by
\begin{align}\label{eq.relAB}
	\bm A(\L_i) &:= \fxd(\Lx) = e^{\fx(\Lx)\tau}\\
	\bm B({\L_j}) &:= \fud(\Lu) = \bigg(\int_{0}^{\tau}e^{\fx(\Lx)t}dt\bigg)\fu(\Lu)\label{eq.relAB2},
\end{align}
where $\tau$ is the sampling period and  $e^{\bm X} = \sum_{k=0}^{\infty}\frac{1}{k!}\bm X^{k}$ is the matrix exponential function~\cite{higham2008functions}. Using then~\eqref{eq.relAB} and~\eqref{eq.relAB2}, we can compactly write model~\eqref{eq.ltia} as
\begin{subequations}\label{eq.lti}
	\begin{eqnarray}
		\xkk &=& \A\xk + \B\uk +\bm w(k)\\
		\yk &=& \C\xk + \D\uk + \bm v(k), 	\label{eq.ltia2}
	\end{eqnarray}
\end{subequations}
where we dropped the dependency of the system matrices $\A$ and $\B$ from $\Lx$ and $\Lu$ to simplify the notation. Throughout the paper, we assume that the matrices $\C$ and $\D$ are known, i.e., we know how our observables are related to the states and inputs..

\section{State-Space Identification}\label{sec.state}


The family of subspace state-space system identification methods~\cite{viberg1994subspace} relies on geometrical properties of the dynamical model~\eqref{eq.lti}. By collecting a batch of $\alpha$ different observables $\y(t) \in \mathbb{R}^L$ into the $\alpha L-$dimensional vector
$$\ybak \triangleq [\bm y(k)^T,\,\ldots,\, \bm y(k+\al-1)^T]^T,$$
we get the relationship
\begin{equation}
    \ybak = \Gammaa\xk + \Phia\ubak + \nbak,
    \label{eq.ltia4n}
\end{equation}
where
\begin{equation}
    \Gammaa \triangleq \begin{bmatrix}
        \C \\
        \C\A\\
        \vdots\\
        \C\A^{\alpha-1}
    \end{bmatrix},
    \label{eq.gamma}
\end{equation}
is the \emph{extended observability} matrix of the system~\cite{ogata2002modern} and
\begin{equation}\label{eq.InputMat}
	\Phia \triangleq \begin{bmatrix}
	\D & 	\bm 0 	& \bm 0 & \cdots& \cdots & \bm 0 \\
	\C\B & 	\D 		& \ddots &	& \cdots & \bm 0 \\
	\vdots & \ddots & \ddots & &  & \vdots\\
	\vdots & & &  & & \bm 0\\
	\C\A^{\al-2}\B 	& \C\A^{\al-3}\B &\cdots & \cdots & \C\B & \D
	\end{bmatrix},
\end{equation}
is the matrix that relates the batch input vector
$$\ubak \triangleq [\bm u(k)^T,\,\ldots,\, \bm u(k+\al-1)^T]^T, $$
with the batch observables $\ybak$. The vector $\nbak$ comprises the batch noise that depends on the system perturbation $\{\bm w(k)\}$ and on the observable noise $\{\bm v(k)\}$. The detailed structure of $\nbak$ is unnecessary for our framework. The size of the batch $\al$ is a user-specified integer and must be larger than the number of states. Assuming the number of nodes is the number of states, this implies $\al > N$.

Given then expression \eqref{eq.ltia4n} and the structures for $\Gammaa$ in \eqref{eq.gamma} and $\Phia$ in \eqref{eq.InputMat}, we proceed by estimating first $\A$ using the algebraic properties of \eqref{eq.gamma} and subsequently $\B$ from the structure of \eqref{eq.InputMat} and a least squares problem.

\vskip2mm
\noindent\textbf{Retrieving the state matrix $\A$.}
A basic requirement for estimating $\A$ is system \emph{observability}~\cite{ogata2002modern}. Observability allows to infer the system state from the outputs for any initial state and sequence of input vectors. Put differently, we can estimate the entire system dynamics from input-output relations. System \eqref{eq.lti} is observable if the system matrices $\{\A,\C\}$ satisfy ${\rm rank}(\mathcal{O}_N) = N$.

Consider now a set of $Q \triangleq  T + \alpha-1$ input-output pairs $\{\yb(k),\bm{u}(k)\}_{k=1}^{Q}$. By stacking the discrete batch vectors $\ybak$ for all observations into the matrix
$$\bm Y = [\yb_{1,\alpha},\,\ldots,\yb_{T,\alpha}]\label{eq.Ymtx},$$
and using expression \eqref{eq.ltia4n}, we can generate the Hankel-structured data equation~\cite{verhaegen2007filtering}
\begin{equation}
    \bm Y = \Gammaa\bm X + \Phia\bm U + \bm N,
    \label{eq.hankel}
\end{equation}
where $\bm X$ is the matrix that contains the evolution of the states accross the columns, i.e.,
$$    \bm X = [\xb(1),\,\ldots,\,\xb(T)],$$
and where the input $\bm U$ and noise $\bm N$ are block Hankel matrices defined as
$$ \bm U = [\bm{u}_{1,\alpha},\,\ldots,\bm{u}_{T,\alpha}] \quad \text{and}\quad \bm N = [\bm{n}_{1,\alpha},\,\ldots,\bm{n}_{T,\alpha}]. $$

The structure in~\eqref{eq.hankel} is at the core of system identification methods~\cite{viberg1995subspace} and this arrangement leads naturally to a subspace-based approach to find $\A$. To detail this, consider the noise-free case
\begin{equation}\label{eq:HanNoNoise}
    \bm Y = \Gammaa\bm X + \Phia\bm U.
\end{equation}
Since the control inputs are known and provided that $\U$ is full row-rank, we can project out $\bm U$ from \eqref{eq:HanNoNoise} by right-multiplying $\Yb$ with the projection matrix
\begin{equation}
    \Pu \triangleq \bm I - \bm U^T(\bm U\bm U^T)^{-1}\bm U.
\end{equation}
Since $\Ub\Pu = \bm 0$, this operation leads to
\begin{equation}\label{eq:projUmain}
    \Yb\Pu = \Gammaa\Xb\Pu.
\end{equation}

Under the assumption that the inputs $\Ub$ are \emph{sufficiently exciting}~\cite{verhaegen2007filtering} (the inputs excite all the modes of the system), the matrix $\Xb\Pu$ has full row-rank. This implies
\begin{equation}
    {\rm span}(\Yb\Pu) = {\rm span}(\Gammaa).
    \label{eq.span}
\end{equation}
That is, the signal subspace of the projected observables $\Yb\Pu$ coincides with that of the extended observability matrix. Therefore, ${\rm span}(\Gammaa)$ can be estimated from $\Yb\Pu$ and, subsequently, the system matrix $\A$ by using the block structure of $\Gammaa$ in~\eqref{eq.gamma}. To detail this procedure, consider the \emph{economy-size} singular value decomposition (SVD) of $\Yb\Pu$
\begin{equation}
    \Yb\Pu = (\Yb\Pu)_N = \Wb_{\al,N}\bm\Sigma_N\bm{V}_{\al,N}^T,
\end{equation}
where $(\Yb\Pu)_N$ is the $N$-rank approximation of $\Yb\Pu$ and equality holds because of the full row-rank assumption of $\Xb\Pu$. Then, from condition~\eqref{eq.span}, we have
\begin{equation}
    \Gammaa = \Wb_{\alpha,N}\bm T,
    \label{eq.urel}
\end{equation}
for some invertible similarity transform matrix $\bm T\in\mathbb{R}^{N\times N}$. Given then $\bm C$ known and full column-rank, we can estimate $\bm T$ from the structure of  $\Gammaa$ [cf.~\eqref{eq.gamma}] as
\begin{equation}
    \hat{\bm T} = (\Jo\Wb_{\al,N})^{\dagger}\C,
    \label{eq.Test}
\end{equation}
where $(\cdot)^{\dagger}$ denotes the Moore-Penrose pseudoinverse and $\Jo\Wb_{\al,N}$ denotes the first $L$ rows of $\Wb_{\al,N}$. If $\bm C$ is not full column-rank, the transform matrix $\bm T$ is not unique. We will deal with the non-uniqueness of $\bm T$ in Section~\ref{sec.partial}.

Finally, to get $\A$, we exploit the \emph{shift invariant} structure of $\Gammaa$ w.r.t. $\A$, i.e.,
\begin{equation}
    \bm{J}_{\rm u}\Gammaa\A = \bm{J}_{\rm l}\Gammaa \in\mathbb{Re}^{(\alpha-1)L\times N},
    \label{eq.sinv}
\end{equation}
where $\bm{J}_{\rm u}\Gammaa$ and $\bm{J}_{\rm l}\Gammaa$ denote the upper and lower $(\al-1)L$ blocks of $\Gammaa$ [cf.~\eqref{eq.gamma}], respectively. By substituting then the expression \eqref{eq.urel} for $\Gammaa$ into \eqref{eq.sinv} and by using the estimate $\hat{\bm T}$ \eqref{eq.Test} for the transform matrix $\bm T$, the least squares estimate for $\A$ is given by
\begin{equation}
    \hat{\A} = (\bm J_{\rm u}\Wb_{\al,N}\hat{\bm T})^{\dagger}\bm J_{\rm l}\Wb_{\al,N}\hat{\bm T}.
    \label{eq.Aest}
\end{equation}

\vskip2mm
\noindent\textbf{Retrieving the input matrix $\B$.}
While the input matrix $\B$ can be obtained following a similar approach as for $\bm A$~\cite{verahegen1992subspace} (yet with a more involved shift-invariant structure), we here compute it together with the initial state $\x(0)$ by solving a least squares problem. 

To do so, we expand system \eqref{eq.lti} to all its terms as
\begin{equation}\label{eq.expYk}
	\begin{aligned}
	\yk{} - (\bm u(k)^T \otimes \bm I_L)&{\rm vec}(\D) =\\ \C\A^k\x(0) + &\bigg(\sum\limits_{q=0}^{k-1}\bm u(q)^T\otimes \C\A^{k-q-1}\bigg){\rm vec}(\B), 
	\end{aligned}
\end{equation}
where $\bm I_L$ is the $L\times L$ identity matrix and vec$(\cdot)$ is the vectorization operator. We then collect the unknowns $\bm x(0)$ and ${\rm vec}(\bm B)$ into the vector $\bm\theta = [\bm x(0)^T\,{\rm vec}(\bm B)^T]^T$ and define the matrix
$$\hat{\bm\Psi} \triangleq \big[\C\hat{\A}^{k},\; \sum\limits_{q=0}^{k-1}\bm u(q)^T\otimes \C\hat{\A}^{k-q-1}\big],$$
where we substituted the state transition matrix $\bm A$ with its estimate \eqref{eq.Aest} while the other quantities $\C, \D$ and $\u(k)$ are known. Finally, we get the input matrix $\B$ by solving
\begin{equation}
	\begin{array}{ll}
		\underset{\bm\theta}{\min}\frac{1}{Q}\sum_{k=1}^{Q}\Vert \y(k) - \hat{\bm\Psi}\bm\theta\Vert_2^2.
	\end{array}
	\label{eq.Best}  
\end{equation}

Given the system matrices $\{\hat{\A},\hat\B,\C,\D\}$, the state interaction graph $\G_x$ ($\Lx$) and input interaction graph $\G_u$ ($\Lu$) can be obtained by enforcing the constraints derived from the information of the physical process, i.e., the model dynamics. Hence, the network structure depends heavily on the estimate of the subspace span of $\Gammaa$.

\subsection{Noisy setting}\label{subse:instrum}

We discuss here a method for estimating $\A$ and $\B$ with perturbations in the state evolution $\bm x(t)$ and noise in the observables $\bm y(t)$. To tackle these challenges, we lever \emph{instrumental variables}.

Consider the following partition of the observables and input matrix
$$\Yb = [\Yb_{1}^T,\, \Yb_{2}^T]^T \quad\text{and}\quad \Ub = [\Ub_{1}^T,\, \Ub_{2}^T]^T,$$
where $\Yb_1$ (resp. $\Ub_{1}$) and $\Yb_2$ (resp. $\Ub_{2}$) have respectively $\beta$ and $\gamma$ blocks of size $L$ with $\gamma = \alpha - \beta$. The model for the observables $\Yb_{2}$ is [cf.\eqref{eq.hankel}]
\begin{equation}
    \bm Y_2 = \Gammag\bm X_2 + \mathcal{T}_\gamma\bm U_2 + \bm N_2.
    \label{eq.hankel2}
\end{equation}
Following then the projection-based strategy to remove the dependency from $\bm U_2$, we can write the projected noisy observables as [cf. \eqref{eq:projUmain}]
\begin{equation}
    \Yb_2\Puf = \Gammag\Xb_2\Puf + \Nb_2\Puf,
    \label{eq.noisyproj}
\end{equation}
with $\Xb_2$ and $\Nb_2$ being the respective partitions of the state evolution and perturbation matrix.

Observe from \eqref{eq.noisyproj} that the signal subspace is corrupted with the noise projection term $\Nb_2\Puf$. To remove the latter, we follow the instrumental variable approach. Since the noise is uncorrelated with the input $\U_1$ and since the noise present in the second block $\Nb_2$ is uncorrelated with the noise in the first block $\Nb_1$, it holds that
 \begin{align}\label{eq.instVar}
     \lim\limits_{T\rightarrow\infty}\frac{1}{T}\Nb_2[\bm U_1^T,\, \Yb_1^T] &= [\bm 0,\,\bm 0].
 \end{align}
Subsequently, we introduce the variable $\bm Z_{1} \triangleq [\Ub_1^T,\,\Yb_1^T]$ and consider the matrix
 \begin{equation}
     \bm G_{1} \triangleq \frac{1}{N}\Yb_2\Puf\bm Z_{1},
     \label{eq.gmtx}
 \end{equation}
to do the estimation of the signal subspace. The matrix $\bm G_{1}$ is asymptotically ``noise free" due to \eqref{eq.instVar}. From the economy-size SVD ($N$-rank approximation) of $\bm G_{1}$, we get
 \begin{equation}
     \bm G_{1} \approx \Wb_{\ga,N}\bm\Sigma_N\bm V_{\ga,N}^T.
     \label{eq.gsvd}
 \end{equation}
Finally, by using $\Wb_{\gamma,N}$, the system matrices $\A$ and $\B$ can be estimated from expressions~\eqref{eq.Test},~\eqref{eq.Aest}, and~\eqref{eq.Best}.

Note that the estimator for the signal subspace [cf.\eqref{eq.gsvd}] has intrinsic statistical properties. To keep under control the variance and bias, different works~\cite{van1995unifying} have proposed to left and right weigh the matrices in~\eqref{eq.gmtx} before the SVD. As establishing optimal weighting matrices requires further analysis, we do not detail them here and refer the interested reader to~\cite{van2012subspace}.

\subsection{Continuous-time model identification}

Given the discrete-time system matrices $\{\hat{\A},\hat{\B},\C,\D\}$, we can estimate the continuous-time transition matrices of~\eqref{eq.lds}. Observe from \eqref{eq.relAB} that estimating $\hat{f}_x(\Lx)$ from $\hat{\A}$ requires computing only the matrix logarithm of $\hat{\A}$. Nevertheless, as stated by the following proposition, there are conditions process $\P$ should meet for this matrix logarithm to $(i)$ exist and $(ii)$ be unique. For clarity, the involved matrices are \emph{real} since we have matrix functions that map real matrices onto real matrices.
\begin{proposition}
Let the analytic function $f_{\rm s,x}(\cdot)$ of $f_x(\Lx)$ in \eqref{eq.relAB} satisfy
\begin{itemize}
		\item[(a)] $e^{f_{\rm s,x}(z)} \not\in \mathbb{R}^{-}, \; \forall \; z \in  {\rm eig}(\Lx)$
		\item[(b)] $f_{\rm s,x}(z) > -\infty, \; \forall \; z \in {\rm eig}(\Lx)$
\end{itemize}
where $\mathbb{R}^{-}$ is the closed negative real axis. Then, process $\mathcal{P}$ guarantees that $(i)$ and $(ii)$ are met.
\label{prop.cond}
\end{proposition}

If the conditions of Proposition~\ref{prop.cond} are met for $f_{{\rm s},x}(\cdot)$, then $\A$ has no eigenvalues on $\mathbb{R}^{-}$. This implies that the principal logarithm\footnote{For two matrices $\bm X$ and $\bm Y$, $\bm X$ is said to be the matrix logarithm of $\bm Y$ if $e^{\bm X} = \bm Y$. If a matrix in invertible and has no-negative real eigenvalues, there is a unique logarithm which is called principal logarithm~\cite{higham2008functions}.} of $\A$, $\ln(\A)~=~ \tau f_x(\Lx)$, exists and is unique. Note that if $\A$ is real, its principal logarithm is also real. Hence, if $f_{{\rm s},x}(\cdot)$ satisfies the conditions of Proposition~\ref{prop.cond}, the continuous-time transition matrices for nonsingular $\hat\A - \bm I$ are given by
\begin{subequations}
	\begin{align}
		\hat{f}_x(\Lx) =& \frac{1}{\tau}\ln(\hat{\A})\\
		\hat{f}_u(\Lu) =& (\hat{\A} - \bm I)^{-1}\hat{f}_x(\Lx)\hat{\B}\label{eq.gEst}.
	\end{align}
\end{subequations}
The expression for $\hat{f}_u(\Lu)$ is derived from
\begin{equation}
	\int\limits_{0}^{\tau}e^{\fx(\Lx)t}dt = \fx(\Lx)^{-1}(e^{\fx(\Lx)\tau}- \bm I).
\end{equation}

Given $\hat{f}_x(\Lx)$ and $\hat{f}_u(\Lu)$, we are left with the estimation of the underlying topologies. The network topology can be estimated with methods that promote sparse representations~\cite{segarra2017network,coutino2018sparsest}, covariance methods~\cite{dempster1972covariance,kumar2019unified} or by \emph{solving a root finding} problem, i.e., inverse maps. We analyzed this last case in the shorter version of this work \cite{coutino2019ITA}. 

However, for the case of partial observations, the similarity transform can not be uniquely identified making the above methods not applicable for identifying the network topology due to the ambiguities in the system matrices. Therefore, in the next section, we discuss such ambiguities and introduce a AP method for estimating the network topology.



\begin{myremark}
    Although we focused principally on continuous-time models, all stated results hold also for purely discrete models by doing the appropriate minor changes for the functional dependencies of the system matrices.
\end{myremark}

\section{Partially Observed Network} \label{sec.partial}


So far, we considered the observables $\y(t)$ are available for the whole network, i.e., $\bm C = \bm I$ or more generally ${\rm rank}(\bm C) = N$. The latter allows a unique estimate for the similarity matrix $\bm T$ in \eqref{eq.Test}. 
We here move to the more involved case where it is not possible to observe the process on all nodes or the observations are not sufficiently rich, i.e., ${\rm rank}(\C) < N$. This disables to find a unique transform matrix $\bm T$ and, instead of retrieving the original system matrices, we retrieve an estimate of a set of equivalent matrices
\begin{align}
				\A_{\rm T} \triangleq \bm T\bm A\bm T^{-1},\;\;\;
    \B_{\rm T} \triangleq \bm T\B,
    \label{eq.At}
\end{align}
which also \emph{realize} the system in \eqref{eq.ltia4n}. Further, it follows from \eqref{eq.At} that (although $\A_{\rm T}\neq \A$ in general) if $\A$ is diagonalizable, i.e., $\A = \bm Q_{A}\bm\Lambda_{A}\bm Q_{A}$, then
\begin{equation}
    {\rm eig}(\A) = {\rm eig}(\A_{\rm T}),
    \label{eq.eigA}
\end{equation}
holds, where ${\rm eig}(\A)$ are the eigenvalues of $\A$. The equality~\eqref{eq.eigA} yields since $\A$ and $\A_{\rm T}$ are \emph{similar} matrices\footnote{Two matrices $\bm X$ and $\bm Y$ are said to be similar if there exists an invertible matrix $\bm P$ such that $\bm X = \bm P \bm Y \bm P^{-1}$.}.

In these situations, we cannot remove the ambiguity in the system matrices without additional information. In the sequel, we motivate why this disambiguation problem is particularly hard. We further derive a method to estimate an \emph{approximately feasible} realization of the \emph{network topology} related to the signal subspace and to the knowledge of the (bijective) scalar mappings $\{f_{{\rm s},x}(\cdot),f_{{\rm s},u}(\cdot)\}$. 

Finally, we stress that despite that in the following part we will concentrate in recovering the state network topology $\mathcal{G}_x$, this is not a problem for cases when the input network topology $\mathcal{G}_u$ is of interest. For such cases, an additional step can always be performed to retrieve $\bm S_u$ from the transformed matrix $\bm B_T$. For instance, from the estimate of $\bm A$ based on the state network topology and $\bm A_T$, the transform $\hat{\bm T}$ can be estimated. Subsequently, the inverse operation is applied to $\B_T$ to get $\hat{\bm B}$ and, therefore, $\mathcal{G}_u$. Hence, there is no detriment on focusing in recovering $\mathcal{G}_x$.


\subsection{The graph inverse eigenvalue problem}

To start, consider the shift operator $\Lx$ belongs to a set $\mathcal{S}$ which contains all permissible matrices representing the dynamics of the state that lead to $\A_{\rm T}$. The set $\mathcal{S}$ describes the properties of the graph representation matrix, e.g., zero diagonal (adjacency) $[\Lx]_{n,n} = 0\,\forall\,n\in [N] $, unitary diagonal (normalized Laplacian) $[\Lx]_{n,n} = 1\,\forall\,n\in [N] $, zero eigenvalue related to the constant eigenvector (combinatorial Laplacian) $\Lx\bm 1 = \bm 0$, symmetry (undirected graphs) $\Lx = \Lx^T$.

The ambiguity \eqref{eq.eigA} introduced by the similarity transform, $\bm T$, transforms the problem of finding $\Lx$ into
\begin{equation}
    \begin{array}{ll}
         \text{find} & \L \\
         \text{subject to} &  \L\in \mathcal{S}\\
            & {\rm eig}(\L) = \lambdabx,
    \end{array}
    \label{eq.feasi1}
\end{equation}
where $\lambdabx = {\rm eig}(\Lx)$ is the vector containing the eigenvalues of $\Lx$ obtained by applying the inverse map to the eigenvalues of $\bm A_T$. Problem~\eqref{eq.feasi1} recasts the network topology identification problem to the problem of finding a graph shift operator matrix $\Lx$ that has a fixed spectrum. This problem belongs to the family of \emph{inverse eigenvalue problems}~\cite{chu2005inverse}. In a way, problem \eqref{eq.feasi1} is the complement of the network spectral template approach~\cite{segarra2017network}. Here, instead of having an eigenbasis and searching for a set of eigenvalues, we have a set of eigenvalues and search for an eigenbasis. 

Problem \eqref{eq.feasi1} is ill-posed as, in most cases, its solution is non-unique. Thus, it leads to  ambiguities in its solution. In what follows, we characterize this ambiguity in terms of equivalence classes between graphs and provide a method able to find a network topology satisfying the conditions of~\eqref{eq.feasi1}.

\begin{figure}
    \centering
    \includegraphics[width=0.35\textwidth]{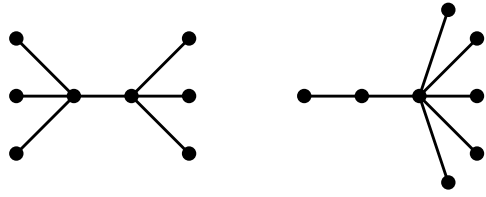}
    \caption{Two cospectral trees with the same number of edges. With respect to the adjacency matrix, almost all trees are non determined by their spectrum. Both graphs have the same characterisitic polynomial $t^4(t^4 - 7t^2 + 9)$, hence cospectral.}
    \label{fig.cospectral}
\end{figure}


\subsection{Ambiguous graphs: cospectral graphs}\label{subsec:ambig}

We find ambiguities in graph topologies in graphs that belong to an equivalence class~\cite{devlin2003sets}. An example of equivalent graphs, in terms of their spectrum, are the isomorphic graphs~\cite{chung1997spectral}. For unlabeled graphs, two graphs $\mathcal{G}$ and $\mathcal{G}^\prime$ with respective graph shift operator matrices $\L$ and $\L^\prime$ are equivalent if there exists a permutation matrix $\bm P$ such that
\begin{equation}
    \bm S = \bm P\bm S^\prime\bm P^T.
\end{equation}
That is, the graph representation matrices are row- and column-permuted versions of each other. The permutation matrix $\bm P$ implements the isomorphism. Therefore, if the only information available is the graph eigenvalues $\bm \lambda$, the graphs are always indistinguishable up to node reordering. This situation is not at all undesirable as the ordering of the nodes is often not important. However, isomorphic graphs are not only ones sharing the spectrum.

Graphs that share the spectrum are called \emph{cospectral} (or isospectral) graphs~\cite{brouwer2011spectra}. Note, however, that cospectral graphs are not necessarily isomorphic. Fig.~\ref{fig.cospectral} illustrates an example of two cospectral graphs. Graph cospectrallity renders the feasible set of~\eqref{eq.feasi1} not a singleton and, therefore, we need to settle with any feasible graph satisfying the constraints. That is, the identified topology from~\eqref{eq.feasi1} will be a graph within the equivalence class of cospectral graphs regarding $\bm\lambda$ and $\mathcal{S}$. The following definition formalizes the latter.

\begin{mydef}\emph{(Equivalent cospectral graphs)} Two graphs $\mathcal{G}$ and $\mathcal{G}^\prime$ of $N$ nodes are cospectral equivalents
with respect to the spectrum $\bm \lambda$ and the graph representation set $\mathcal{S}$, if they belong to the set
$$
    \mathcal{C}_{\mathcal{S}}^{\bm\lambda} := \{ \mathcal{G}\, |\, \Sb \in \mathcal{S},\,{\rm eig}(\bm S) = \bm\lambda \}.
$$
\end{mydef}

As a result, the feasibility problem~\eqref{eq.feasi1} reduces to a \emph{graph construction} problem (graph inverse eigenvalue problem). Put differently, the problem at hand can be rephrased as \emph{given a spectrum $\bm\lambda$ and a set $\mathcal{S}$, construct a graph shift operator matrix $\bm S \in \mathcal{S}$ with spectrum $\bm\lambda$}. We shall discuss next a method that addresses this construction problem.

\subsection{Graph construction by alternating projections}

Before detailing the graph construction method, we introduce the following assumptions.

\begin{enumerate}
    \item [(A.1)] The set $\mathcal{S}$ is closed.
    \item [(A.2)] For any $\bm S\in\mathbb{R}^{N\times N}$, the projection $P_{\mathcal{S}}(\Sb)$ of $\bm S$ onto the set $\mathcal{S}$ is unique.
\end{enumerate}
The first assumption is technical and guarantees that the set $\mathcal{S}$ includes all its limit points. The second assumption is slightly more restrictive and ensures that the problem
\begin{equation}\label{eq.projS}
    \begin{array}{lll}
         P_{\mathcal{S}}(\Sb) := &\underset{\hat{\bm S}}{{\rm minimize}} & \Vert \Sb - \hat{\Sb}\Vert_{\rm F},\;\text{s.t.}\; \hat{\Sb} \in \mathcal{S},
    \end{array}
\end{equation}
has a unique solution. Although this assumption might seem restrictive, in most cases we only have access to a convex description of the feasible set $\mathcal{S}$ which satisfies A.2 (as the set is assumed closed) or only the projection onto the convex approximation of the feasible set of network matrices can be performed efficiently. Thus, it is fair to consider that A.2 holds in practice.

For the sake of exposition, we restrict our upcoming discussion to the case of symmetric matrices, i.e., undirected graphs but remark that a similar approach can be followed for directed graphs\footnote{This could be done by exchanging the spectral decomposition for the Schur decomposition~\cite{golub2012matrix} which decomposes a matrix into unitary matrices and an upper triangular matrix.}. Further, denote by $\mathcal{S}^{N}$ the set of symmetric $N\times N$ matrices and by $\mathcal{S}_{+}^{N}$ the set of positive semidefinite matrices. We then recall the following result from~\cite[Thm. 5.1]{chu1990projected}.

\def\bigO{{\rm O}}
\def\Qb{\bm Q}
\begin{theorem}[adapted] Given $\bm S \in\mathcal{S}^{N}$ and let $\Sb = \Qb \Lambdab\Qb^T$ be the spectral decomposition of $\Sb$ with non-increasing eigenvalues $[\Lambdab]_{ii} \geq [\Lambdab]_{jj}$ for $i < j$. For a fixed $\Lambdab_o\in\mathcal{D}_N$ with non-increasing elements $[\Lambdab_o]_{ii} \geq [\Lambdab_o]_{jj}$ for $i < j$, a best approximant, in the Frobenius norm sense, of $\Sb$ in the set
$$
    \mathcal{M} := \{ \bm M \in \mathcal{S}^{N}\, |\, \bm M = \bm V \bm\Lambda_o \bm V^T, \bm V \in \bigO(N) \},
$$
is given by
\begin{align*}
    P_{\mathcal{M}}(\Sb) := \Qb\Lambdab_o\Qb^T,
\end{align*}
where $\bigO(N)$ denotes the set of the $N\times N$ orthogonal matrices.
\label{th.spectral} 
\end{theorem}

Theorem~\ref{th.spectral} implies that the projection of $\L$ onto $\mathcal{M}$ is not necessarily \emph{unique}. As an example, consider the graph shift operator $\bm S$ with repeated eigenvalues. Here, it is not possible to uniquely define a basis for the directions related to the eigenvalues with multiplicity larger than one. Hence, infinitely many eigendecompositions exist that lead to many projections of $\bm S$ onto $\mathcal{M}$. Further, note that as every element of $\mathcal{M}$ is uniquely determined by an element of $\bigO(N)$, the structure of $\mathcal{M}$ is completely defined by the structure of $\bigO(N)$. Therefore, as $\bigO(N)$ is a \emph{smooth manifold}, $\mathcal{M}$ is one as well.

\vskip2mm
\noindent\textbf{Alternating projections method.} As it follows from Assumptions A.1 and A.2 and Theorem~\ref{th.spectral}, we can project any graph shift operator matrix $\bm S \in \mathcal{S}^{N}$ onto $\mathcal{S}$ and $\mathcal{M}$. Further, by noticing that the construction problem~\eqref{eq.feasi1} is equivalent to finding a matrix in
\begin{equation}
    \mathcal{S} \cap \mathcal{M},
    \label{eq.inter}
\end{equation}
we can consider the simple, and intuitive, method of \emph{alternating projections} (AP)~\cite{bauschke1996projection} to find a point in~\eqref{eq.inter}. The AP method finds a point in the intersection of two \emph{closed convex} sets by iteratively projecting a point onto the two sets. It performs the updates
\begin{subequations}\label{eq:APmethod}
    \begin{eqnarray}
        \Sb_{k + 1/2} =& P_{\mathcal{S}}(\Sb_{k})\\
        \Sb_{k + 1} \in& P_{\mathcal{M}}(\Sb_{k+1/2}), 
    \end{eqnarray}        
    \label{eq.it}
\end{subequations}
\hspace{-0.8ex}starting from a point $\Sb_0 \in \mathcal{M}$. The AP method has guaranteed convergence for convex sets and it does that linearly. However, for alternating projections on a combination of different types of sets (we have a set $\mathcal{S}$ satisfying A.1 and A.2, and a smooth manifold $\mathcal{M}$), additional conditions on both sets are necessary to guarantee convergence.

First, let us formalize the notion of a fixed point for the iterative procedure~\eqref{eq.it}.


\begin{mydef}[Fixed point]\label{def.fixPoint} A matrix $\bm S \in \mathcal{S}^{N}$ is a fixed point of the alternating projections procedure in~\eqref{eq.it} if there exists an eigendecomposition of $P_{\mathcal{S}}(\Sb)$,
$$
    P_{\mathcal{S}}(\Sb) = \Qb\Lambdab\Qb^T \in \mathcal{S},
$$
with non-increasing elements $[\Lambdab]_{ii} \geq [\Lambdab_{jj}]$ if $i < j$ such that
$$
    \Sb = \Qb \Lambdab_o \Qb^T \in \mathcal{M}.
$$
\end{mydef}

This definition makes explicit two things. First, it defines $\bm S$ as a fixed point if and only if
\begin{equation}
    \Sb \in P_{\mathcal{M}}(P_\mathcal{S}(\Sb)).
\end{equation}
Second, for the cases where $\Sb$ is a fixed point and $P_{\mathcal{S}}(\Sb)$ has eigenvalues with multiplicity larger than one, progress can be still made towards a feasible solution when $P_{\mathcal{M}}(P_\mathcal{S}(\Sb))\not\in\mathcal{S}\cap\mathcal{M}$. To see the latter, consider the case where an alternative eigendecomposition
\begin{equation}
    P_{\mathcal{S}}(\Sb) = \tilde{\Qb}\Lambdab\tilde{\Qb}^T,
\end{equation}
is available for the fixed point $\Sb$. Assuming that
\begin{equation}
    \tilde{\Sb} = \tilde{\Qb}\Lambdab_o\tilde{\Qb} \neq \Qb \Lambdab_o \Qb^T = \Sb,
\end{equation}
with $\Qb$ the eigenbasis that makes $\Sb$ a fixed point, we can see that the new point $\tilde{\Sb}$ escapes from the fixed point. Further, since successive projections between two closed sets is a nonincreasing function over the iterations~\cite[Thm. 2.3]{orsi2006numerical}, we can prove that $\tilde{\Sb}$ presents a progress towards a feasible solution in $\mathcal{S}\cap\mathcal{M}$.\footnote{
 We considered the notion of fixed point to obtain a feasible set of the system matrices (matrices that realize the system) since, beyond their structure, the most important characteristic is their spectrum (set of eigenvalues).} 
 
With this in place, the following theorem shows that the AP method converges when used for the graph construction problem.


\begin{theorem}\label{th.alter}
Let $\mathcal{S}$ meet Assumptions A.1-A.2 and consider the set $\mathcal{M}$ defined in Theorem~\ref{th.spectral}. Let also $\Sb_0,\Sb_1,\Sb_2,\ldots,$ be a sequence generated by the alternating projections method in~\eqref{eq.it}. Then, there exists a limit point $\Sb$ of this sequence that is a fixed point of~\eqref{eq.it} [cf. \ref{def.fixPoint}] satisfying
$$
    \Vert \Sb - P_{\mathcal{S}}(\Sb) \Vert = \lim_{k\rightarrow\infty}\Vert\Sb_k - P_{\mathcal{S}}(\Sb_k)\Vert.
$$
Further, if the limit is zero, then $\Sb\in\mathcal{S}\cap\mathcal{M}$.
\end{theorem}
\begin{proof}
    See Appendix~\ref{ap.th2}\footnote{Appendix in the supplemental material.}.
\end{proof}

The above theorem proves that the AP method retrieves a matrix $\Sb\in\mathcal{M}$ that realizes the original system, i.e., it preserves the underlying system modes and is an approximately \emph{feasible network representation}. Nevertheless, Theorem~\ref{th.alter} does not quantify the rate of convergence of such a method. By particularizing results for super-regular sets~\cite[Thm. 5.17]{lewis2007local}, the following theorem shows that if the problem is feasible, locally, the proposed method converges linearly to a point in~\eqref{eq.inter}.


\begin{theorem}\label{th.th3}
Let the set of all permissible matrices $\mathcal{S}$ [cf. \eqref{eq.feasi1}] be convex and meet Assumptions A.1-A.2. Let also the set $\mathcal{M}$ be defined as in Theorem~\ref{th.spectral}. Denote by $N_{\mathcal{S}}({\Sb})$ the normal cone of the closed set $\mathcal{S}$ at a point $\Sb$ and, similarly, by $N_{\mathcal{M}}({\Sb})$ the normal cone of the set $\mathcal{M}$ at $\bm S$. Further, suppose that $\mathcal{M}$ and $\mathcal{S}$ have a strongly regular intersection at $\bar{\Sb}$, i.e., the constant
$$
   \bar{c} = \max\{\langle u , v \rangle\,:\, u\in N_{\mathcal{M}}(\bar{\Sb})\cap B,\,v \in -N_{\mathcal{S}}(\bar{\bm S})\cap  B\},
$$
is strictly less than one with $B$ being a closed unit Euclidean ball. Then, for any initial point $\Sb_0 \in \mathcal{M}$ close to $\bar{\Sb}$, any sequence generated by the alternating projections method in~\eqref{eq.it} converges to a point in $\mathcal{M}\cap\mathcal{S}$ with $R$-linear rate 
$$
    r \in ({\bar{c}},1).
$$
\end{theorem}
\begin{proof}
    See Appendix~\ref{ap.th3}.
\end{proof}

These results guarantee that the AP method converges \emph{globally} (at least) to a fixed point in $\mathcal{M}$ and \emph{locally}, i.e., within the neighborhood of the solution (if it exists), to a fixed point in $\mathcal{M}\cap\mathcal{S}$.

\subsection{Inaccurate and partial information}

We now consider the case where the estimated eigenvalues are inexact because of noise or are incomplete because the full eigendecomposition of the system matrices is not feasible. To deal with such cases, we modify the structure of the set $\mathcal{M}$ (in Theorem~\ref{th.spectral}) to reflect the uncertainty and the partial eigendecomposition. The modified set has to be compatible with the structure used in Definition~\ref{def.fixPoint} and Theorems~\ref{th.alter} and \ref{th.th3} to guarantee the convergence of the AP method. Thus, in the sequel, we focus on proving \emph{compactness} for the modified versions of $\mathcal{M}$, which suffices to guarantee convergence of the AP method by the result of Theorem~\ref{th.alter}.

\vskip2mm
\noindent\textbf{Uncertainty in the system matrices.} The following proposition shows that the set $\mathcal{M}_\epsilon$, which allows the estimated eigenvalues to lie within an $\epsilon-$uncertainty ball, is compact.

\begin{proposition}\label{pr.p2}
Let $\Lambdab_o \in \mathcal{D}_N$ with $[\Lambdab_o]_{ii} \geq [\Lambdab_o]_{jj}$ for $i < j$ be fixed. If $0 \leq \epsilon < \infty$ is a fixed scalar accounting for uncertainties on the elements of $\Lambdab_o$, then the set
\begin{align}\label{eq.setep}
    \begin{split}
    \mathcal{M}_{\epsilon} 
    := \{\bm M \in \mathcal{S}^N~|~\bm M = \bm V ( \Lambdab_{o} &+ \Lambdab_{\epsilon})\bm V^T, \,\bm V \in \bigO(N),\,\\
    &\Lambdab_{\epsilon}\in\mathcal{D}_{N}, \Vert \Lambdab_{\epsilon} \Vert_2 \leq \epsilon \}
    \end{split}
\end{align}
is compact.

\end{proposition}
\begin{proof}
    See Appendix~\ref{ap.pro2}.
\end{proof}


The following result provides a best approximant of a matrix $\bm S$ in the set $\mathcal{M}_\epsilon$ in the Frobenius norm.

\begin{theorem}\label{th.noise} Given $\bm S \in\mathcal{S}^{N}$ with eigendecomposition $\Sb = \Qb \Lambdab\Qb^T$ and non-increasing eigenvalues $[\Lambdab]_{ii} \geq [\Lambdab]_{jj}$ for $i < j$. For a fixed $\Lambdab_o\in\mathcal{D}_N$ with $[\Lambdab_o]_{ii} \geq [\Lambdab_o]_{jj}$ for $i < j$, a best approximant of $\Sb$ in $\mathcal{M}_{\epsilon}$, in the Frobenius norm sense, is given by
$$
    P_{\mathcal{M}_{\epsilon}}(\Sb) := \Qb(\Lambdab_o + \Lambdab_{\epsilon}^*)\Qb^T,
$$
where
$$
    \Lambdab_{\epsilon}^* := \underset{\Lambdab_{\epsilon}\in\mathcal{D}_{N}}{{\rm argmin}}\; \Vert \Lambdab - \Lambdab_{o} - \Lambdab_\epsilon \Vert_{{\rm F}},\; {\rm s.t}\; \Vert \Lambdab_{\epsilon}\Vert_2 \leq \epsilon.
$$
\end{theorem}
\begin{proof}
    See Appendix~\ref{ap.th4}.
\end{proof}


\begin{mycoro} The nonzero entries of $\Lambdab_\epsilon^*$ are
$$
    [\Lambdab_{\epsilon}^*]_{ii} = {\rm sign}(\gamma_i)\cdot{\rm min}\{\epsilon, |\gamma_i|\},
$$
with $\gamma_i := [\Lambdab]_{ii} - [\Lambdab_{o}]_{ii}$.
\end{mycoro}

\vskip2mm
\noindent\textbf{Partial eigendecomposition.} In physical systems, a discrete model of $N$ degrees of freedom provides accurate information of about $N/3$ of the system natural frequencies~\cite[Ch. 5]{chu2005inverse}. In other cases, the full eigendecomposition of the system matrix is not always possible. We, therefore, provide a projection onto a set that only considers a noisy part of the system matrix spectrum is available.

The following theorem provides the main result.



\begin{theorem}\label{th.th5} Let $\Lambdab_m\in \mathcal{D}_{m}$ with $[\Lambdab_m]_{ii} \geq [\Lambdab_m]_{jj}$ for $i < j$ be fixed. If $ 0 \leq \epsilon < \infty$ is a fixed scalar accounting for uncertainties on the elements of $\Lambdab_{m}$ and $\rho := \max_{\L \in \mathcal{S}} \Vert \L \Vert_2$, then a best approximant, in the Frobenius norm sense, of $\L\in\mathcal{S}^{N}$, with $\Vert \L \Vert_2 \leq \rho$, in the set
\begin{equation}
    \begin{split}
    &\mathcal{M}_{\epsilon}^m := \{ \bm M \in\mathcal{S}^N\, |\, \bm M = \bm V {\rm bdiag}(\Lambdab_{m} + \Lambdab_{\epsilon},\bar{\Lambdab})\bm V^T,\\ &\bm V \in \bigO(N),\,\Lambdab_{\epsilon}\in\mathcal{D}_{m},\, \Vert \Lambdab_{\epsilon} \Vert_2 \leq \epsilon, \bar\Lambdab\in\mathcal{D}_{N-m}, \Vert\bar{\Lambdab}\Vert_2 \leq \rho \}
    \end{split}
\end{equation}
is given by
$$
    P_{\mathcal{M}_{\epsilon}^m}(\Sb) := \Qb{\rm bdiag}(\Lambdab_m + \Lambdab_{\epsilon}^*,\Lambdab_{\bar{\bm\sigma}})\Qb^T.
$$
Here, $\bm\sigma$ denotes the permutation of the subset of $[N]$ that solves the combinatorics problem
\begin{equation}\label{eq.permInd}
    \underset{1\leq [\bm\sigma]_1<\ldots<[\bm\sigma]_m\leq N}{{\rm min}}\sum\limits_{i=1}^{m}([\Lambdab]_{[\bm\sigma]_i[\bm\sigma]_i} - [\Lambdab_m]_{ii})^2,
\end{equation}
where $\Lambdab$ is the diagonal matrix of eigenvalues of $\bm S$, and $\bar{\bm\sigma}$ is the complementary set of $\bm\sigma$. The matrix $\Qb$ is given by the (sorted) eigendecomposition of $\Sb$, i.e.,
$$
    \Sb = \Qb{\rm bdiag}(\Lambdab_{\bm\sigma},\Lambdab_{\bar{\bm \sigma}})\Qb^T,
$$
where $\Lambdab_{\bm \sigma}$ is the permuted version of $\Lambdab$
and
$$
    \Lambdab_{\epsilon}^* := \underset{\Lambdab_{\epsilon}\in\mathcal{D}_{m}}{{\rm argmin}}\; \Vert \Lambdab_{\bm\sigma} - \Lambdab_{m} - \Lambdab_\epsilon \Vert_{{\rm F}},\;{\rm s.t.}\Vert\Lambdab_{\epsilon}\Vert_2\leq \epsilon.
$$
Furthermore, $\mathcal{M}_{\epsilon}^m$ is compact.
\end{theorem}
\begin{proof}
    See Appendix~\ref{ap.th5}.
\end{proof}

\begin{mycoro}
The optimal permutation $\bm \sigma$ of the indices $[N]$ that solves \eqref{eq.permInd} can be found by solving a minimum-weight bipartite perfect matching problem.
\end{mycoro}

The above results in Theorems~\ref{th.noise} and \ref{th.th5} showed that the sets $\mathcal{M}_\epsilon$ and $\mathcal{M}_\epsilon^m$ are compact and provide a best Frobenius-norm approximant for $\bm S$ in each case. Therefore, we can apply the AP method in \eqref{eq.it} to these scenarios using the appropriate modifications. Finally, since the sets $\mathcal{M}_\epsilon$ and $\mathcal{M}_\epsilon^m$ meet the conditions of Theorem~\ref{th.alter}, the convergence results for the AP method extend also to the discussed scenarios.

\section{System Consistency Constraints}\label{sec.consistency}
The set $\mathcal{S}$ [cf. \eqref{eq.feasi1}] plays an important role in the system topology that the AP method identifies. As such, it should be constrained such that the AP method yields a consistent system, i.e., the AP estimated state network should define an equivalent system to the original one. We briefly discuss here two constraints that can be added to $\mathcal{S}$ to enforce system consistency.

Given that the system matrix [cf.~\eqref{eq.relAB}] is a bijective matrix function and by using the same construction as for the shift invariance property in~\eqref{eq.sinv}, we can build the linear system
\begin{equation}\label{eq.const1}
    \begin{bmatrix}
    \bm C_{T}\\
    \bm C_{T}f_{x}^{-1}(\bm A_T)
    \end{bmatrix}\bm T = \begin{bmatrix}
    \bm C \\ \bm C\Lx
    \end{bmatrix}.
\end{equation}
Here, we leverage the invariance of the matrix function to nonsingular transforms, i.e.,
\begin{equation}
    f(\bm A_T) = \bm T f(\bm A)\bm T^{-1},
\end{equation}
where $f(\cdot)$ is a matrix function, hence, can be applied to $\bm A_T$. In this way, we get a linear system that depends on $\Lx$ and enforces the shift invariance condition. Nevertheless, the shift invariance condition does not change the optimality nor the uniqueness of the projection onto $\mathcal{S}$. This is because $\bm T$ is a free optimization variable and does not affect the projection distance [cf.~\eqref{eq.projS}]. 

If other constraints for the transform matrix $\bm T$ are known, they can be included when solving for the projection~\eqref{eq.projS}. These additional constraints will not impact the projection optimality because they do not change the cost function for $P_{\mathcal{S}}$. For instance, consider the constraint that requires symmetry in the input matrix $\bm B$. Then, since $\bm B_T = \bm T\bm B$, $\bm B$ is symmetric, and $\bm B_T \bm T = \bm T \bm B \bm T^T$, we have
\begin{equation}
    \bm B_T\bm T^T = \bm T\bm B_T^T.
\end{equation}
Again, such an additional constraint does not change the projection distance as it only modifies the description of the convex set in which the matrix must be projected.

We can introduce other constraints to the set $\mathcal{S}$ to further restrict the family of feasible network representations. However, these constraints should be analyzed as case-specifc and go beyond the main goal of this work. Next, we corroborate the above theoretical findings with numerical results.

\section{Numerical Results}\label{sec.numerical}
In this section, we present a series of numerical results to illustrate the performance of the proposed methods for different scenarios. We first illustrate how the model and the noise coloring influence the estimation performance of commonly used topology identification methods. Then, we corroborate our theoretical results and, finally, we present results for the topology identification from partially-observed networks.\footnote{The code to generate these numerical results can be found in \url{https://gitlab.com/fruzti/systemid_codes}}


\subsection{Discrete model validation}

In this section, we corroborate the discrete model \eqref{eq.ltia} in finding a graph topology from continuous-time data generated following the model \eqref{eq.lds}. The underlying graphs are two fixed random regular graphs of $N = 50$ nodes with node degree of $d = 3$. The data are generated by a continuous-time solver with system evolution matrix $f_x(\Lx) = -\Lx$ and input matrix $f_u(\Lu) = -(\Lu + \bm I)$. The observable matrix $\bm C$ is set to identity and $\bm D$ is the zero matrix. The input signal is drawn from a standard normal distribution and we set the number of samples to $N^3$ with a sampling time of $\tau = 10^{-3}$. 

We compare the proposed method with the spectral templates techniques in \cite{segarra2017network}. For the latter, the system matrices $\bm A$ and $\bm B$ are first obtained from the continuously sampled data (cf. Section~\ref{sec.model}). Then, the eigenvectors of these matrices are used as spectral templates. These results are shown in Fig.~\ref{fig.model}.

Fig.~\eqref{fig.modelGraph} shows the estimated network topologies for a particular input signal realization, while Fig.~\eqref{fig.modelEig} compares the respective spectra. We observe that the approach relying on spectral templates overestimates the number of edges and underestimate the graph eigenvalues. However, the graph obtained with spectral templates is a \emph{matrix function} of the original graph, i.e., there is a function (polynomial) that maps the estimated graph to the original one. This is because $\Lx$ has all eigenvalues with multiplicity one. The proposed technique relying on the discrete model \eqref{eq.ltia} retrieves the eigenvalues and the graph structure perfectly. This result is not surprising since subspace-based system identification is a consistent estimator for the transition matrix and the proposed method use the knowledge of $f_{s,x}$ while the spectral templates does not. For this scenario, we also considered building the graph from the data covariance matrix, but this technique did not lead to satisfactory results. We attribute this misbehavior to the fact that the covariance matrix is not diagonalizable by the graph modes (i.e., eigenvectors of $\Lx$) due to the prescence of the input.

\begin{figure*}
     \centering
    \begin{subfigure}{.6\textwidth}
    \centering
    \includegraphics[width=\textwidth]{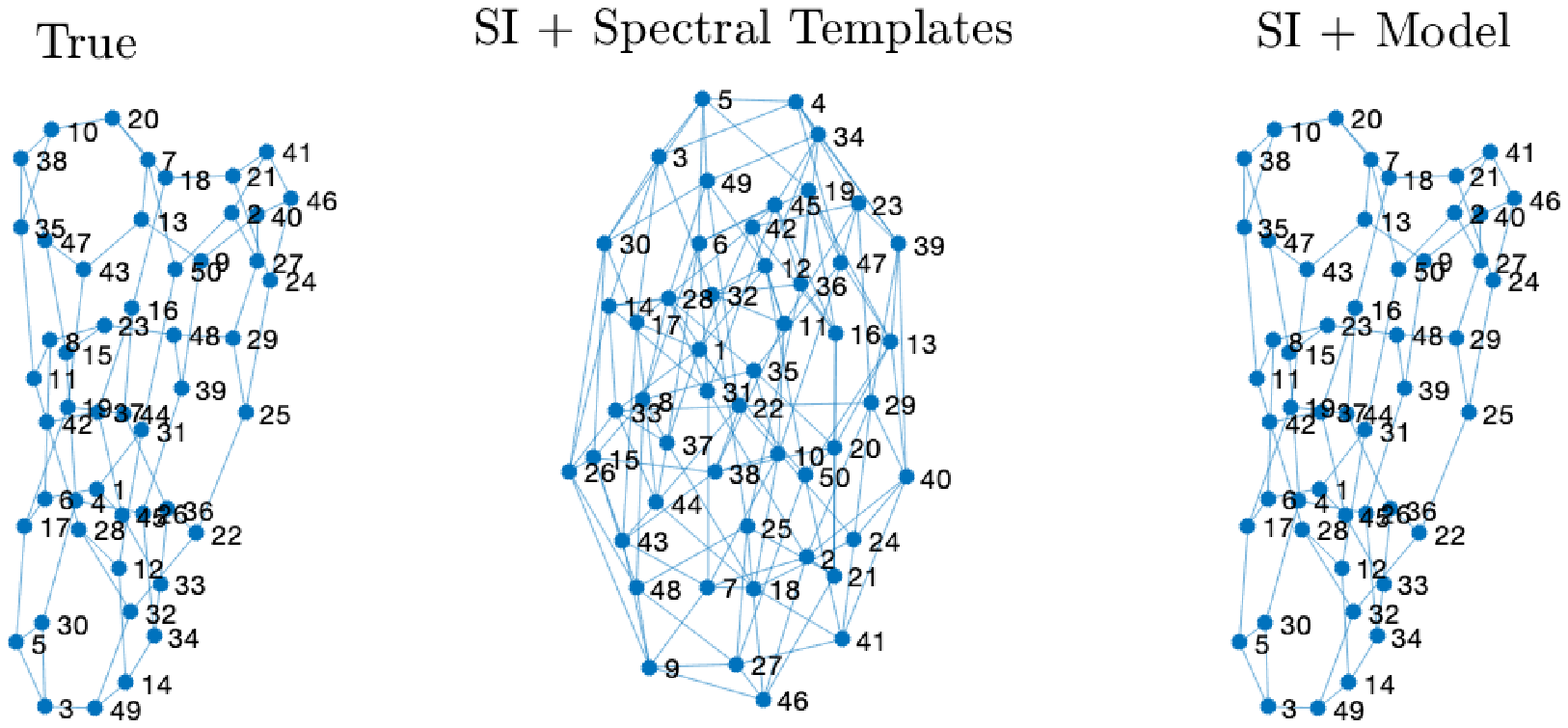}
    \caption{}
    \label{fig.modelGraph}
    \end{subfigure}%
    \begin{subfigure}{.4\textwidth}
    \centering
    \psfrag{lambda}{$\lambda_k$}
    \psfrag{[k]}{$[k]$}
    \psfrag{True}{\tiny True}
    \psfrag{SI + Model}{\tiny SI + Model}
    \psfrag{Si + Spectral Templates}{\tiny SI + Spectral Templates}
    \includegraphics[width=\textwidth]{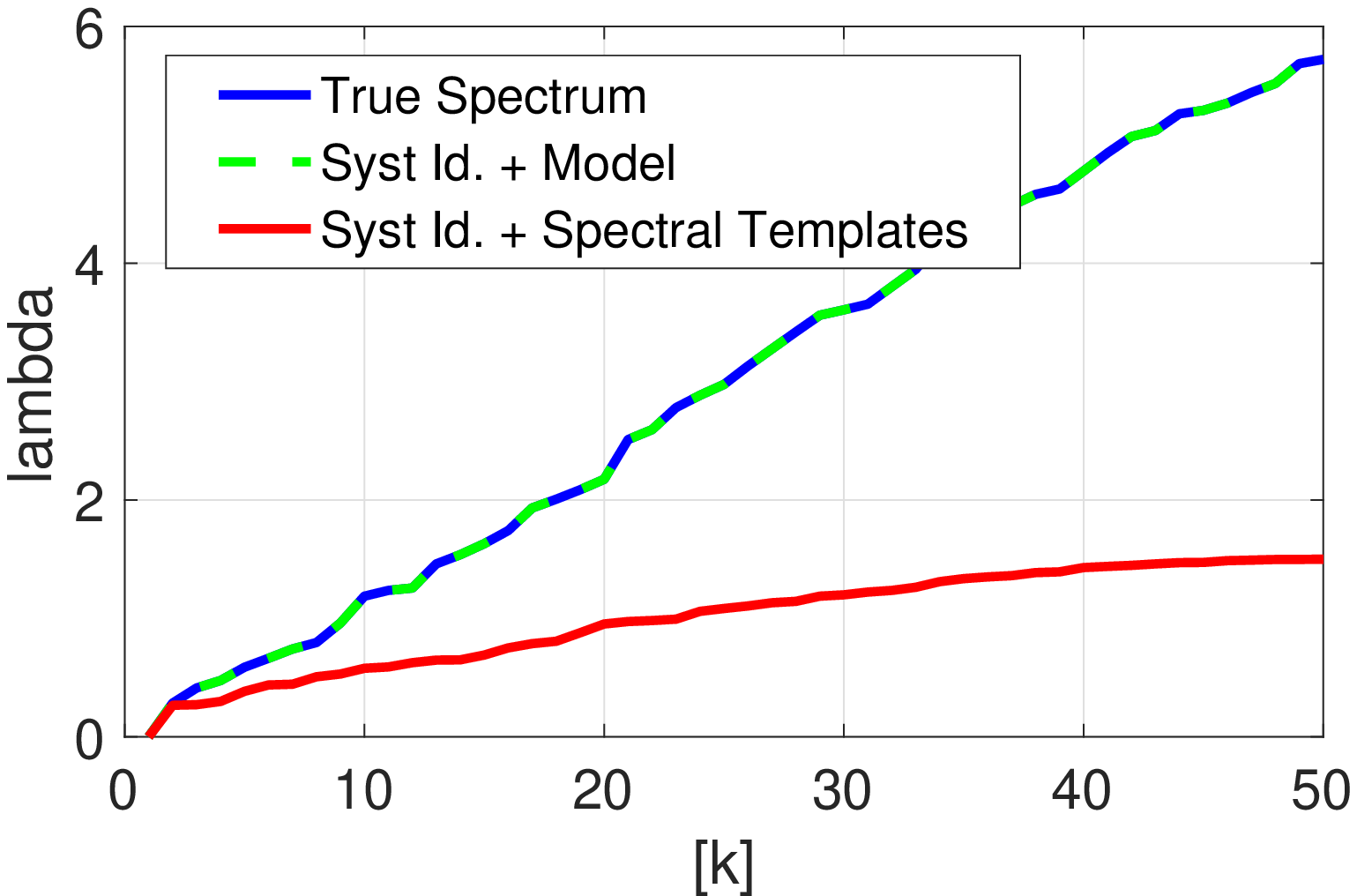}
    \caption{}
    \label{fig.modelEig}
    \end{subfigure}
    \caption{Comparison of  the spectral template method within the system identification framework and the model-aware method. (a) Reconstructed state graph. (b) Comparison of eigenvalues of the estimated graphs.}
    \label{fig.model}
\end{figure*}


\subsection{Instrumental variables approach on social graph}

We now evaluate the instrumental variable approach of Section~\ref{subse:instrum} on the Karate club graph~\cite{zachary1977information}. The graph represents the connections of $N = 34$ members through $78$ undirected edges. We consider the discrete model~\eqref{eq.ltia} with $\bm A$ described by a continuous-time diffusion process $\bm A = e^{-\tau\bm L_{x}}$. The diffusion rate (or sampling time) is fixed to $\tau = 10^{-3}$. The input signal is randomly generated from a standard normal distribution and the power of both the state and the observation noise is $\sigma^2 = 10^{-3}$

We consider three different approaches to estimate the underlying network topology: $i)$ a covariance-based approach, where the covariance matrix is estimated from the observables; $ii)$ the instrumental variable approach combined with the spectral template method from \cite{segarra2017network}; and $iii)$ the proposed instrumental variable approach by enforcing the dynamics of the continuous system. These results are reported in Fig.~\ref{fig.instVar}.

Fig.~\ref{fig.ivVec} shows the fitting accuracy of the subspaces, while Fig.~\ref{fig.ivVal} illustrates the fitting of the eigenvalues. In Fig.~\ref{fig.ivGraphs}, we provide the obtained graphs where the edges with absolute weight less than $10^{-3}$ are omitted. We observe that the system identification flow allows a better graph reconstruction and the proposed method offer the best alignment of the eigenbasis. Further, note that by levering the underlying physical model of the diffusion, we can reconstruct the graph spectrum with high fidelity. We further remark that despite both the basis and the spectrum are aligned, the retrieved graph looks different from the true one. This is because of the ambiguities discussed in Section~\ref{subsec:ambig}. However, the obtained graph has the same eigenvalues as the original one and its basis diagonalizes the original network matrix. Notice also that from the three methods, only the one leveraging the model information retrieves a connected graph after thresholding. 

Finally, we remark that the task of estimating this topology based purely on a spectral decomposition is hard. This is because the combinatorial Laplacian of the Karate club graph has eigenvalues with a multiplicity larger than one. Thus, there is no unique basis for its eigendecomposition leading to difficulties in reconstructing the underlying topology. This issue is also present even when the topology is the sparsest matrix that generates such dynamics.

\begin{figure*}
    \centering
    \begin{subfigure}{.5\textwidth}
    \centering
    \psfrag{Covariance-Based}{\tiny Covariance-Based}
    \psfrag{Spectral Templates}{\tiny IV+Spectral Templates}
    \psfrag{Proposed}{\tiny Proposed}
    \includegraphics[width=\textwidth]{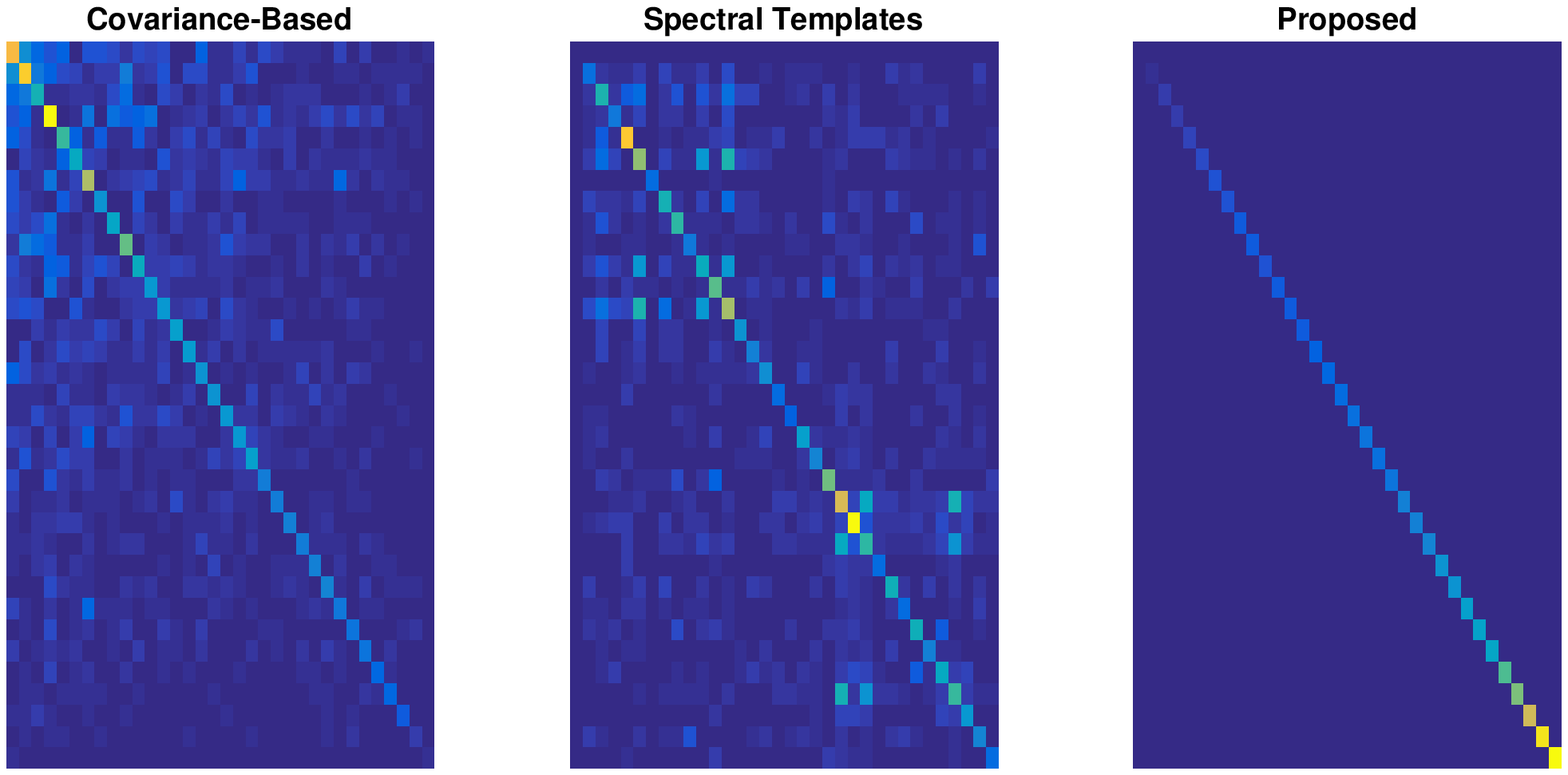}
    \caption{}
    \label{fig.ivVec}
    \end{subfigure}%
    \begin{subfigure}{.5\textwidth}
    \centering
    \psfrag{lambda}{$\lambda_k$}
    \psfrag{[k]}{$[k]$}
    \psfrag{Covariance-Based}{\tiny Covariance-Based}
    \psfrag{Proposed}{\tiny Proposed}
    \psfrag{Spectral Templates}{\tiny IV+Spectral Templates}
    \psfrag{True}{\tiny True}
    \psfrag{Eigenvalues}{Eigenvalues}
    \includegraphics[width=\textwidth]{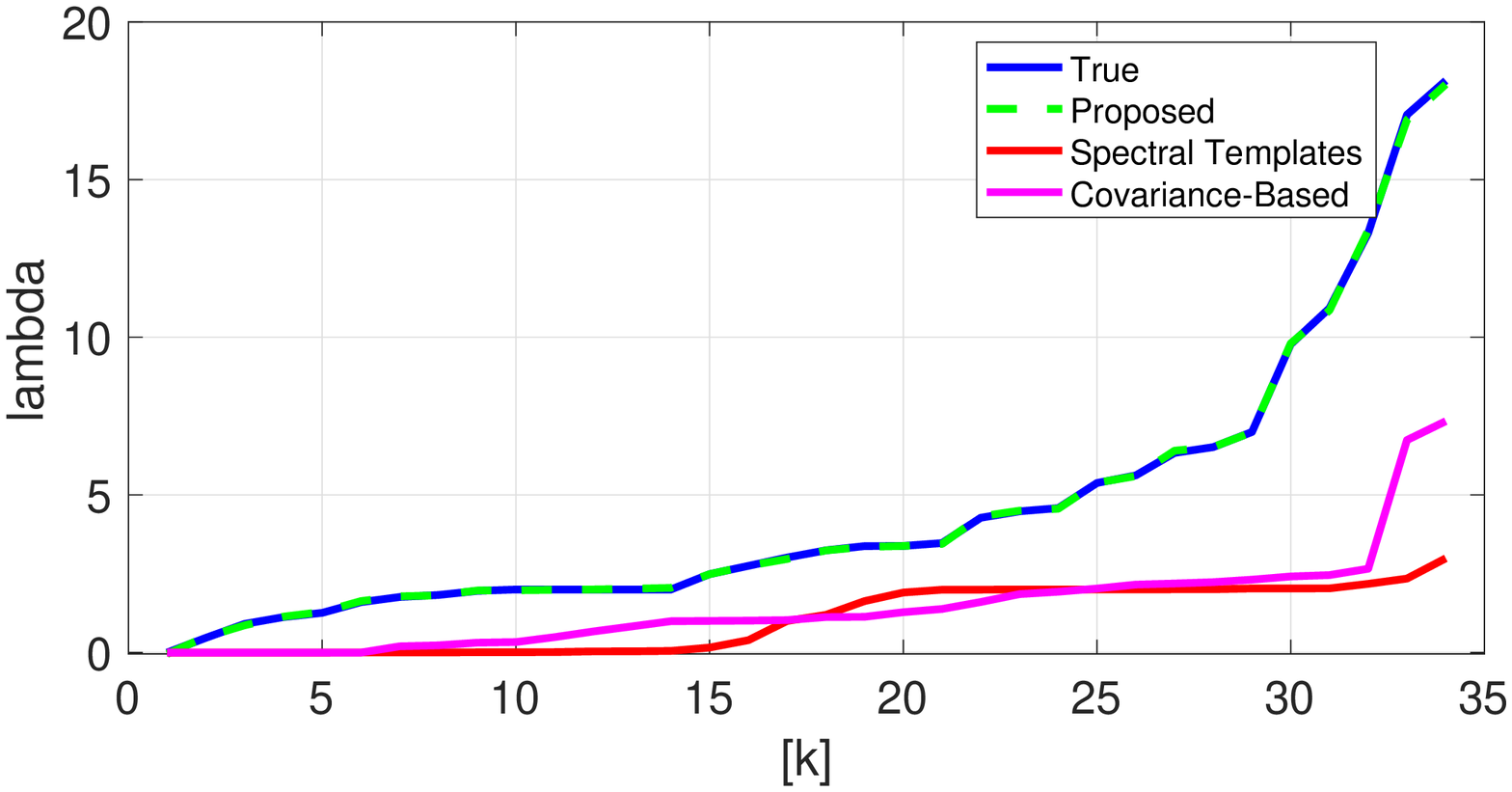}
    \caption{}
    \label{fig.ivVal}
    \end{subfigure}
    \begin{subfigure}{\textwidth}
    \centering
    \includegraphics[width=\textwidth]{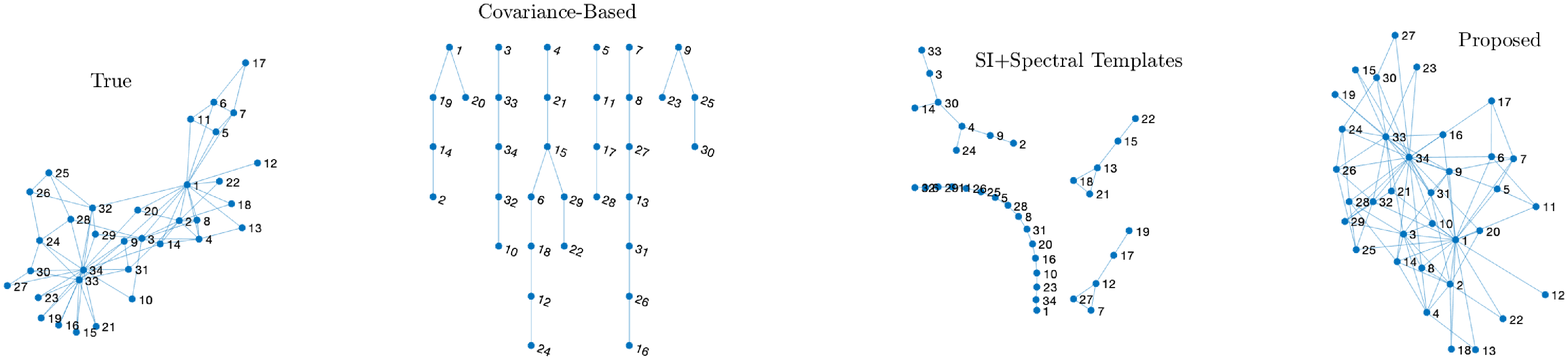}
    \psfrag{Proposed}{\tiny Proposed}
    \psfrag{True}{\tiny True}
    \caption{}
    \label{fig.ivGraphs}
    \end{subfigure}
    \caption{Comparison of several methods using and not using the instrumental variable approach. (a) Comparison of alignment of the eigenbasis of the estimated graphs with the ones of the true graph. (b) Comparison of eigenvalues of estimated graphs. (c) Comparison of estimated topologies.}
    \label{fig.instVar}
\end{figure*}


\subsection{Convergence of the alternating projections method}

We analyze here the convergence behavior of the alternating projections method \eqref{eq:APmethod}. We present results using the sets $\mathcal{M}_{\epsilon}^m$ and $\mathcal{S} = \mathcal{L}_{{\rm CVX}}$. The latter is the convex relaxation of the combinatorial Laplacian set. These sets are chosen to illustrate the convergence results as $\mathcal{L}_{{\rm CVX}}$ encompasses the problem of finding Laplacian matrices with given eigenvalues and $\mathcal{M}_{\epsilon}^m$ is the most general set proposed in this work. Additional results for the other sets are provided in the supplementary material. For this scenario, we select a regular graph\footnote{See supplemental material} with $N=30$ and node degree $d=3$ and consider only half of its eigenvalues known, i.e, $m = N/2$. The AP method is analyzed for five different initial points. 

These results are shown in Fig.~\ref{fig.convResults}. Here, each solid line represents a different starting point. The (blue) dashed line shows the convergence behavior when the starting point is the (diagonal) eigenvalue matrix. These results show two main things. First, the predicted monotone behavior of the error $\Vert \bm S_{k} - P_{\mathcal{S}}(\bm S_k)\Vert_{{\rm F}}$ holds and stagnates when a limit point is reached by the iterative sequence. Second, the error $\Vert \bm S_k - \bm S_{k+1}\Vert_{{\rm F}}$ converges to the desired accuracy (order $10^{-6}$), although not monotonically, The error convergence rate is generally linear and the starting point influences the slope. Finally, we emphasize that even when the set of known eigenvalues lies within an $\epsilon$-ball of uncertainty, the alternating projections method convergences. The convergence is guaranteed by the compactness of the set $\mathcal{M}_{\epsilon}^m$.

\begin{figure*}
    \centering
    \begin{subfigure}{.5\textwidth}
    \centering
    \psfrag{Error}[Bc]{\scriptsize $\;\;\;\;\;\;\;\;\;\;\;\;\;\Vert\bm S_k - P_{\mathcal{L}_{{\rm CVX}}}(\bm S_k)\Vert_{\rm F}$}
    \psfrag{Iteration [k]}[bc]{\scriptsize $\;\;\;\;\;\;\;\;\;\;\;\;\;\;\;\;\;\;\;\text{Iteration}\; [k]$}
    \includegraphics[width=\textwidth]{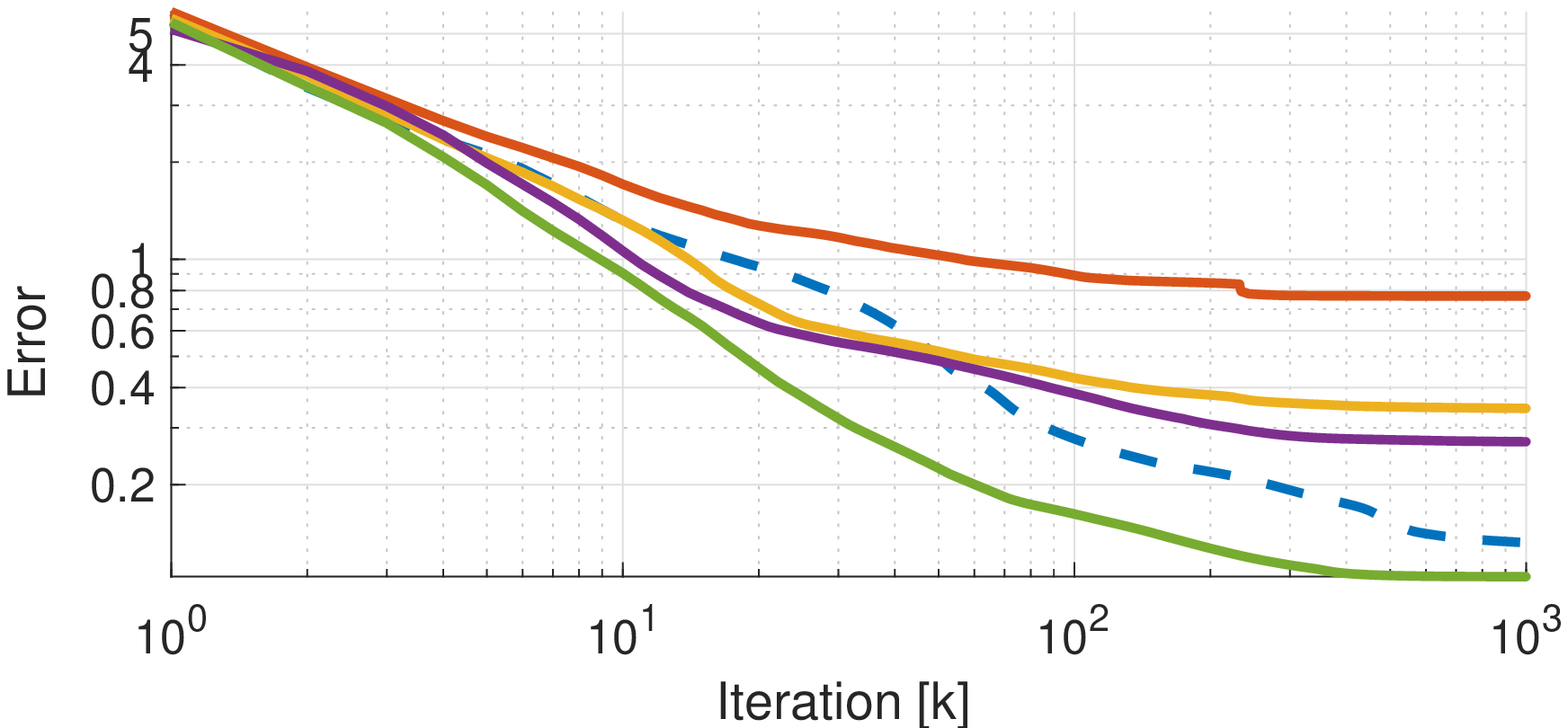}
    \caption{}
    \end{subfigure}%
    \begin{subfigure}{.5\textwidth}
    \centering
    \psfrag{Error}[cb]{\scriptsize $\;\;\;\;\;\Vert\bm S_{k} - \bm S_{k+1}\Vert_{\rm F}$}
    \psfrag{Iteration [k]}[bc]{\scriptsize $\;\;\;\;\;\;\;\;\;\;\;\;\;\;\;\;\;\;\;\text{Iteration}\; [k]$}
    \includegraphics[width=\textwidth]{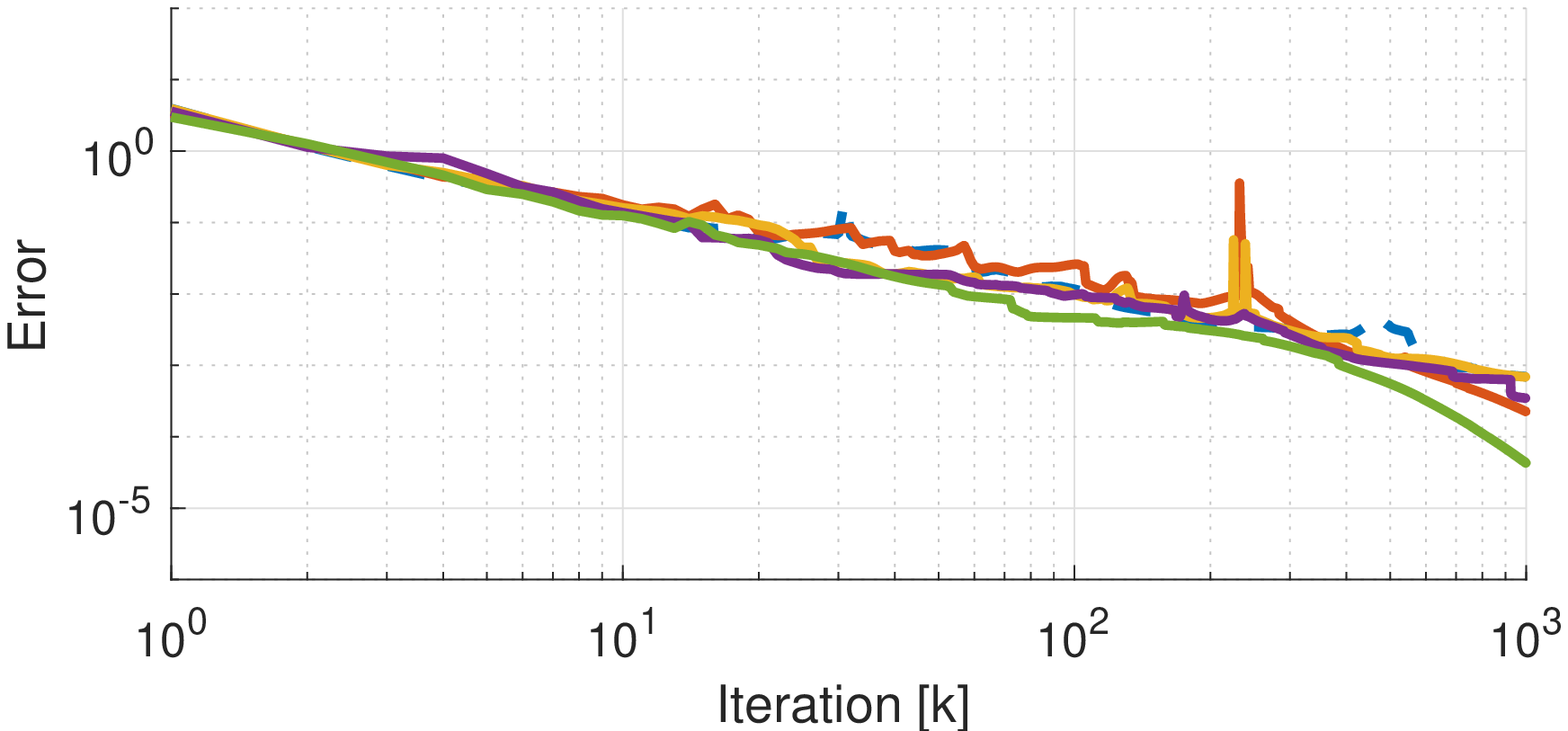}
    \caption{}
    \end{subfigure}
    \caption{Convergence plots for the alternating projections method with $M_{\epsilon}^{m}$ and $\mathcal{L}_{{\rm CVX}}$. (a) Error with respect to the projection, i.e., $\Vert\bm S_k - P_{\mathcal{L}_{\rm CVX}}(\bm S_k)\Vert_{\rm F}$. (b) Iterate error, i.e., $\Vert\bm S_{k} - \bm S_{k+1}\Vert_{\rm F}$}
    \label{fig.convResults}
\end{figure*}


\subsection{Partial observations}

In this section, we consider the task of retrieving a graph that realizes a given system from partial observations. We consider two regular random graphs of $N=14$ nodes and three edges per node. The data are generated from a continuous diffusion on the network and the input matrix is $\bm B = \bm L_{u} + \bm I$. The matrix $\bm C$ is a Boolean matrix that selects half of the nodes (the odd labeled nodes from arbitrary labeling). Note that none of the previous methods can be employed to retrieve the network topology since $\bm B \neq \bm I$ and the network is not fully observed. Even if the covariance matrix is estimated from sampled data, its eigenstructure does not represent the eigenstructure of the state network topology.

We first estimate the system matrices using the system identification framework and then employ the AP method initialized with a random symmetric matrix that has as eigenvalues the estimated state network eigenvalues. The constraint set in \eqref{eq.feasi1} is the convex relaxation of the combinatorial Laplacian set \cite{segarra2017network}. We further enrich this set with the system identification constraints to enforce the feasibility of the realization. Fig.~\ref{fig.numOb} reports the results after $30$ iterations of the AP method.

From Fig.~\ref{fig.nobProj}, we observe that the estimated state graph does not exactly share the eigenbasis with the original one, i.e., the graph mode projections do not form a diagonal matrix. However, we could perfectly match the input graph. This behavior is further seen in the eigenvalues, where those of the input graph are matched by the estimated eigenvalues. For the state graph, a perfect eigenvalue match is possible if a final projection onto $\mathcal{M}$ is performed. These results are also reflected in the estimated topologies in Fig.~\ref{fig.numOb}. Perfect reconstruction of the input graph support is achieved, while the state graph presents a different arrangement in the nodes and it is not regular.

Despite the differences in the state graphs, the triple $\{\hat{\L}_x,\hat{\L}_u,\bm C\}$ realises (approximately) the same system as the triple $\{\Lx,\Lu,\bm C\}$. This is because the product of the involved system matrices is preserved, i.e., although the structure of the state graph is different, the observations can be reproduced with high confidence using the estimated system matrices. With the estimated graphs, we can predict the system output with an NRMSE fitness of $\approx95\%$. In the supplementary material, we compare the response of the true system to an arbitrary excitation with the obtained system response that uses the estimated graphs.

\begin{figure*}
    \centering
    \begin{subfigure}{.5\textwidth}
    \centering
    \psfrag{State Graph}{\scriptsize States Graph}
    \psfrag{Input Graph}{\scriptsize Inputs Graph}
    \includegraphics[width=\textwidth]{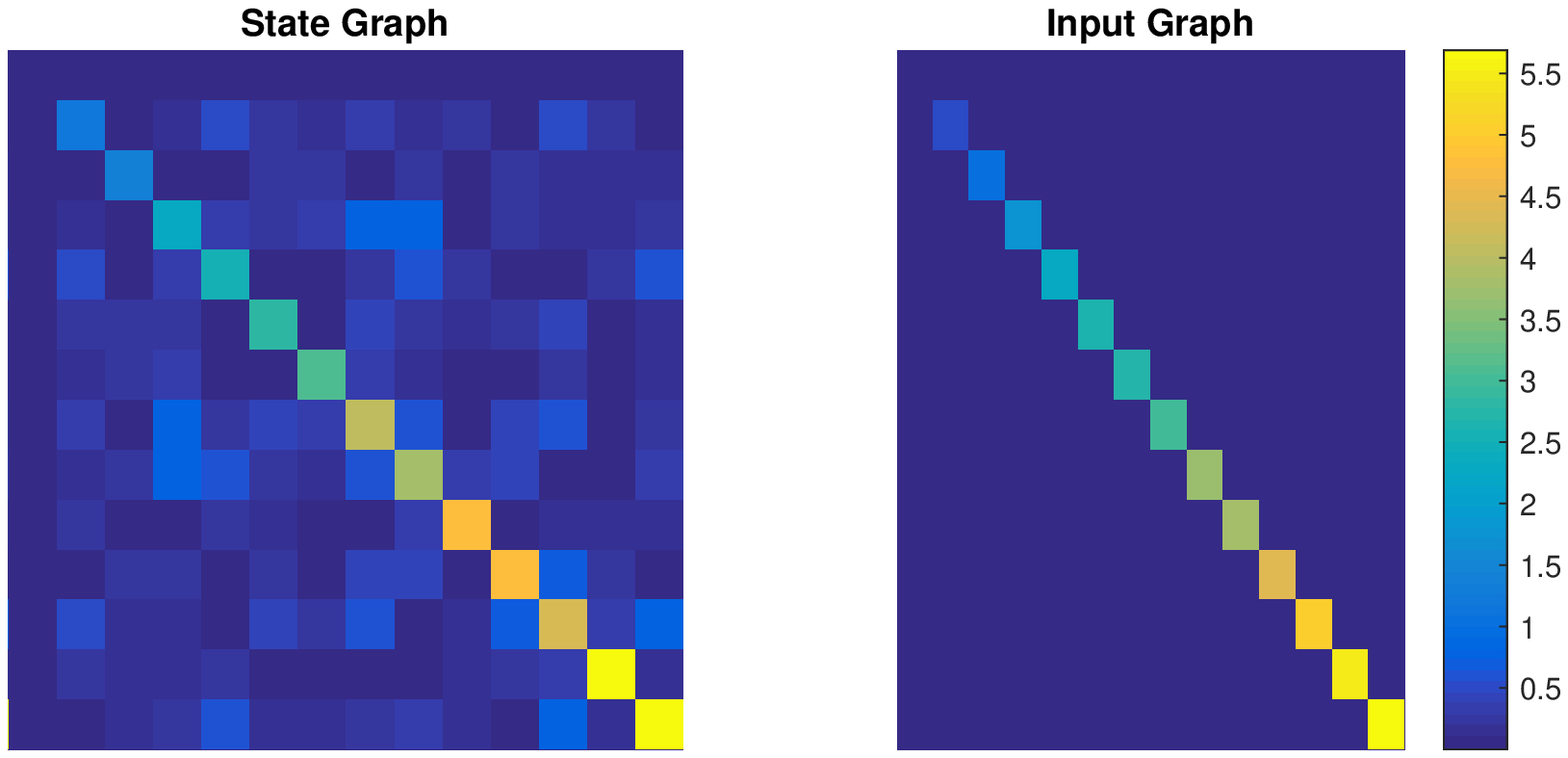}
    \caption{}
    \label{fig.nobProj}
    \end{subfigure}%
    \begin{subfigure}{.5\textwidth}
    \centering
    \psfrag{lambda}[tb]{$\;\;\;\;\;\lambda_k$}
    \psfrag{[k]}{$[k]$}
    \psfrag{Laaaaaa}{\tiny${\rm eig}(\Lx)$}
    \psfrag{Lbaaaaa}{\tiny${\rm eig}(\Lu)$}
    \psfrag{Laestaaaaa}{\tiny${\rm eig}(\hat{\L}_x)$}
    \psfrag{LbEstaaaaa}{\tiny${\rm eig}(\hat{\L}_u)$}
    \includegraphics[width=\textwidth]{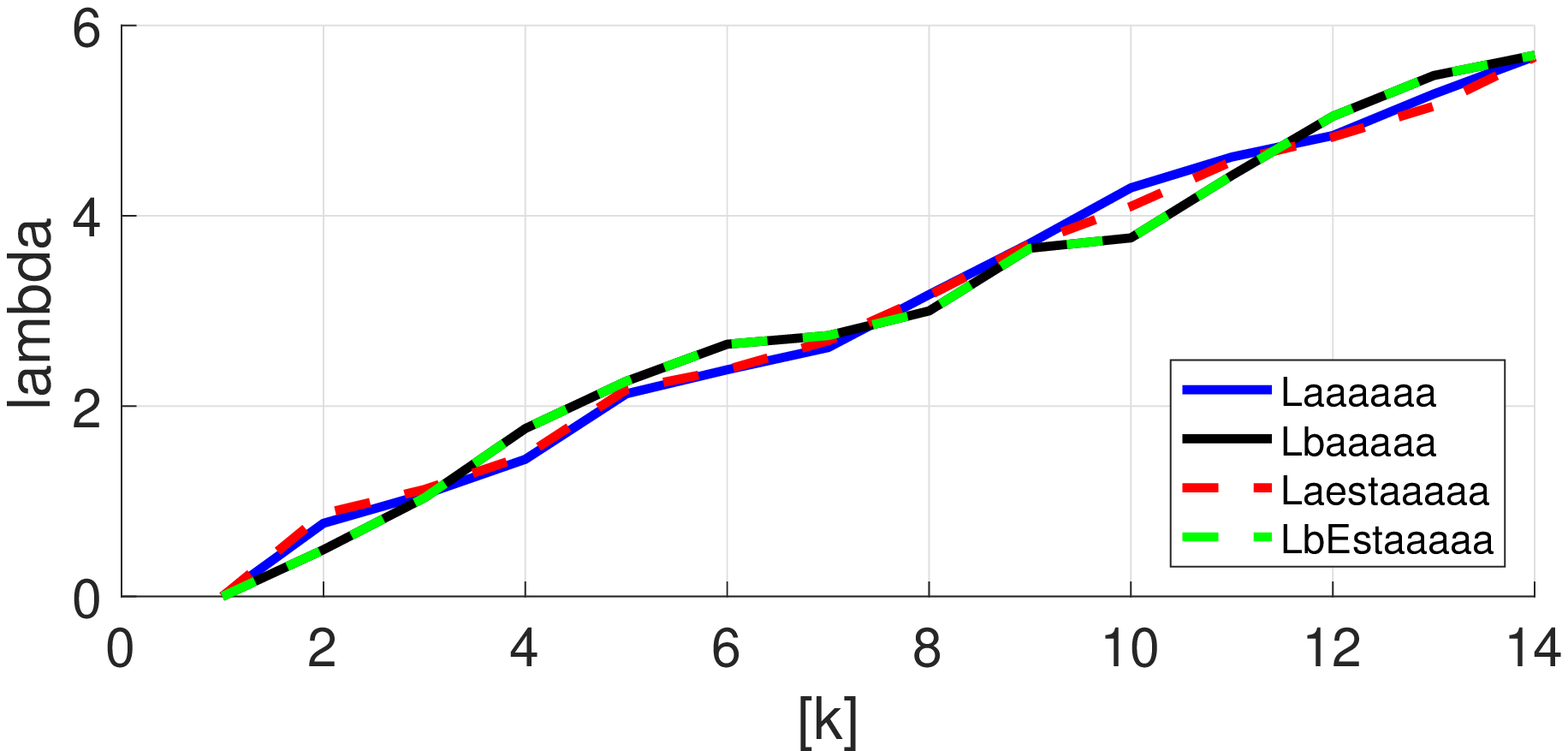}
    \caption{}
    \label{fig.nobEig}
    \end{subfigure}
    \begin{subfigure}{\textwidth}
    \centering
    \includegraphics[width=\textwidth]{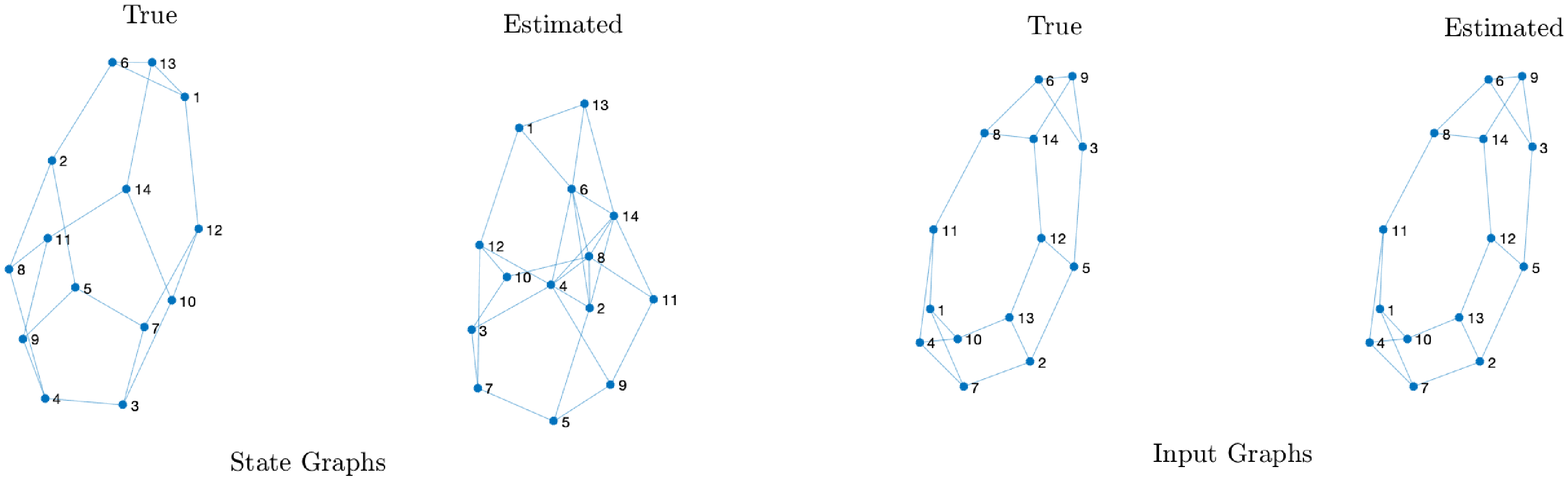}
    \caption{}
    \label{fig.nobGraph}
    \end{subfigure}
    \caption{Results for the reconstruction of a graph from a dynamical system using partial observations and the alternating projections method. (a) Projection results on the modes of the estimated graphs, i.e., $|\hat{\Qb}_*^T\L_*\hat{\Qb}_{*}|$ with $\Qb_*$ being the eigenvectors of the estimated graph. (b) Comparison of the estimated graph eigenvalues. (c) Reconstructed graphs.}
    \label{fig.numOb}
\end{figure*}
\section{Conclusions}\label{sec.conclusions}
This paper introduced a general framework for graph
topology identification through state-space models and subspace techniques.
We showed that it is possible to retrieve the matrix
representation of the involved graphs from the system matrices by exploiting the geometric structure of the input and output data. In particular, we discussed the challenges of retrieving the network topology under partial observations and proposed an alternating projections method to recover a set of matrices that realizes the system and are conform with user-defined constraints. The proposed theoretical analysis is corroborated with numerical results. Future research is needed in three main directions. First, the focus should be on improving the scalability of the proposed techniques to larger graphs. Second, research is needed in employing subspace models to learn a coarser graph that drives the system dynamics in large data sets. Third, extensions of the current approach to nonlinear systems, such as the ones in chemical reaction networks are also worth being investigated.



\bibliographystyle{IEEEtran}
\bibliography{dissertation}

\ifappx

\newpage
\clearpage
\section{Appendix}


\subsection{{Proof Theorem~\ref{th.alter}}}\label{ap.th2}


We construct the proof of the convergence of the method with arguments similar to ~\cite{orsi2006numerical}.

Note that any matrix $\Sb_k \in \mathcal{M}$ can be represented by a matrix $\Qb \in \bigO(N)$. Then, since $\bigO(N)$ is compact, yields the set $\mathcal{M}$ is compact. The latter implies $\mathcal{M}$ contains a convergent sequence and that $\L_k$ has a limit point denoted by $\L^*$. Further, consider the implementation of~\eqref{eq.it}
$$
    P_{\mathcal{S}}(\Qb_k) = \Qb_k \Lambdab_k \Qb_k^T,
$$
and
$$
    \Sb_{k+1} = \Qb_k \Lambdab_o\Qb_k^T.
$$
From the compactness of $\bigO(N)$, $\Qb_k$ converges to a point $\Qb$ that, in turn, makes $\Sb_{k+1}$ converge to $\Sb$.

Further, since the successive projections between $\mathcal{S}$ and $\mathcal{M}$ is a nonincreasing function over the iterations~\cite[Thm. 2.3]{orsi2006numerical}, the limit $\lim_{k\rightarrow\infty}\Vert \Sb_k - P_{\mathcal{S}}(\Sb_k)\Vert$ exists. Then, since $P_{\mathcal{S}}(\cdot)$ is a continuous operation, we have that
\begin{align}\label{eq.equChain}
\begin{split}
    \Vert \Sb^* - P_{\mathcal{S}}(\Sb^*)\Vert &= \lim_{k\rightarrow\infty}\Vert \Sb_{k} - P_{\mathcal{S}}(\Sb_k)\Vert \\
    & = \lim_{k\rightarrow\infty}\Vert \Sb_{k+1} - P_{\mathcal{S}}(\Sb_{k+1})\Vert \\
    & = \Vert \Sb - P_{\mathcal{S}}(\Sb)\Vert.
\end{split}
\end{align}

Then, since $P_{\mathcal{S}}(\Sb_k)$ converges to $P_{\mathcal{S}}(\Sb^*)$ and the orthogonal matrix $\Qb_k$ converges to $\Qb$, it follows that $P_{\mathcal{S}}(\Sb^*) = \Qb\Lambdab^*\Qb^T$ for some $\Lambdab^*$.
Further, the eigendecomposition $\Sb = \Qb\Lambdab_o\Qb^T$ implies that $\Sb = P_{\mathcal{M}}(P_{\mathcal{S}}(\Sb^*))$.

From \eqref{eq.equChain} and since $\Sb$ is the projection of $P_{\mathcal{S}}(\Sb^*)$ onto $\mathcal{M}$, we have
$$
    \Vert \Sb - P_{\mathcal{S}}(\Sb) \Vert \geq \Vert \Sb - P_{\mathcal{S}}(\Sb^*) \Vert.
$$
By considering then that the projection of $\Sb$ onto $\mathcal{S}$ is unique (Assumption A.2), we have $P_{\mathcal{S}}(\Sb^*) = P_{\mathcal{S}}(\Sb)$. Further, since $\Sb$ is a projection of $P_{\mathcal{S}}(\Sb^*)$ onto $\mathcal{M}$, $\Sb$ is a fixed point.

Finally, consider that the limit in \eqref{eq.equChain} is zero. An arbitrary sequence generated by the AP method converges to some point $\tilde{\Sb}\in\mathcal{M}$, which is a limit point. From the inequality
\begin{align*}
 \Vert \tilde{\Sb} - P_{\mathcal{S}}(\Sb_k)  \Vert &= 
 \Vert \tilde{\Sb} - P_{\mathcal{S}}(\Sb_k) + \Sb_k - \Sb_k  \Vert \\ 
 &\leq \Vert \Sb_k - P_{\mathcal{S}}(\Sb_k) \Vert + \Vert \tilde{\Sb} - \Sb_k \Vert \\ 
 &= \Vert \tilde{\Sb} - \Sb_k \Vert,
\end{align*}
and by taking the limit for $k\rightarrow\infty$, we observe that $\tilde{\Sb}$ is also a limit of points in $\mathcal{S}$. As both sets are closed by assumption, the result follows.

\subsection{{Proof Theorem~\ref{th.th3}}}\label{ap.th3}

Consider that $\mathcal{S}$ and $\mathcal{M}$ are a convex set and a smooth manifold, respectively. Therefore, they are both super-regular sets~\cite{lewis2007local}, which implies that $\mathcal{S}$ and $\mathcal{M}$ are super-regular sets at $\bar{\Sb}$. By the results from~\cite[Thm. 5.15]{lewis2007local} and~\cite[Thm. 5.17]{lewis2007local} that guarantee the convergences of alternating projections in super-regular sets (under the same condition at the intersection as in this theorem), the convergence guarantee follows.





\subsection{{Proof Proposition~\ref{pr.p2}}}\label{ap.pro2}
Consider that the set
$$
\mathcal{D}_{N,\epsilon} := \{\Lambdab_{\epsilon}\,|\,\Lambdab_{\epsilon}\in\mathcal{D}_N,\,\Vert \Lambdab_\epsilon \Vert_2 \leq \epsilon\}
$$
is compact. This holds since the set $D_{N,e}$ is equivalent to the set of diagonal matrices whose diagonal entries lie in $[-\epsilon,\epsilon]$. Hence, the compactness.

Consider now the map $\phi: \bigO(N)\times\mathcal{D}_{N,\epsilon}\rightarrow\mathcal{M}_{\epsilon}$ given by
$$
\phi(\bm V,\Lambdab_{\epsilon}) = \bm V (\Lambdab_o + \Lambdab_\epsilon)\bm V^T
$$
which defines the set $\mathcal{M}_{\epsilon}$ and observe that $\bigO(N)$ is compact. The map $\phi$ is continuous and surjective, thus the set $\mathcal{M}_\epsilon$ is compact.





\subsection{{Proof Theorem~\ref{th.noise}}}
\label{ap.th4}

The proof starts by expressing the set $\mathcal{M}_\epsilon$ as
\begin{equation}\label{eq.unionDesc}
    \mathcal{M}_{\epsilon} := \bigcup\limits_{\Lambdab_\epsilon\in\mathcal{D}_{N,\epsilon}} \mathcal{M}(\Lambdab_\epsilon),
\end{equation}
where $\mathcal{M}(\Lambdab_\epsilon)$ is defined similarly to the set $\mathcal{M}$ [cf. Theorem~\ref{th.spectral}] by substituting $\Lambdab_{o}$ with $\Lambdab_o + \Lambdab_\epsilon$. Then, we can expand the projection problem for $\Sb$ as
$$
   \begin{array}{llll}
        \underset{\bm M\in\mathcal{M}_{\epsilon}}{{\rm min}} & \Vert \Sb - \bm M \Vert_{{\rm F}} &
        =\underset{\Lambdab_\epsilon\in\mathcal{D}_{N,\epsilon}}{\rm min}\;\underset{\bm M\in\mathcal{M}(\Lambdab_\epsilon)}{{\rm min}}\Vert \Sb - \bm M \Vert_{{\rm F}}\\
        &&=  \underset{\Lambdab_\epsilon\in\mathcal{D}_{N,\epsilon}}{\rm min}  \Vert \Sb - P_{\mathcal{M}(\Lambdab_\epsilon)}(\Sb) \Vert_{{\rm F}} \\
        &&=  \underset{\Lambdab_\epsilon\in\mathcal{D}_{N,\epsilon}}{\rm argmin}  \Vert \Sb - \bm Q( \Lambdab_o + \Lambdab_\epsilon)\bm Q^T \Vert_{{\rm F}} \\
        && =  \underset{\Lambdab_\epsilon\in\mathcal{D}_{N,\epsilon}}{\rm argmin}  \Vert \Lambdab - \Lambdab_o - \Lambdab_\epsilon \Vert _{{\rm F}}.
   \end{array}
$$
To obtain the first equality, we used the union description of $\mathcal{M}_{\epsilon}$ in \eqref{eq.unionDesc}. The second and third equalities follow from Theorem~\ref{th.spectral}. The last equality follows from the invariance of the Frobenius norm under unitary transforms. Finally, observe that the non-increasing ordering of eigenvalues for both $\Lambdab$ and $\Lambdab_o$ solves the last optimization problem, thus providing the best Frobenius-norm approximant for $\bm S$.

\subsection{Enforcing Positive semidefinitness under noise}

Observe that the set $\mathcal{M}_{\epsilon}$ [cf. Prop.~\ref{pr.p2}]
\begin{align*}
    \begin{split}
    \mathcal{M}_{\epsilon} 
    := \{\bm M \in \mathcal{S}^N~|~\bm M = \bm V ( \Lambdab_{o} &+ \Lambdab_{\epsilon})\bm V^T, \,\bm V \in \bigO(N),\,\\
    &\Lambdab_{\epsilon}\in\mathcal{D}_{N}, \Vert \Lambdab_{\epsilon} \Vert_2 \leq \epsilon \}
    \end{split}
\end{align*}
specifies a neighborhood where the eigenvalues of $\bm M$ must lie. However, if the desired operator $\bm M$ is positive semidefinite (PSD), the constraint in the norm $\Vert \Lambdab_{\epsilon} \Vert_2$ might lead to matrices that are not PSD. The following proposition introduces an alternative set $\mathcal{M}_{\bm a,\bm b}$ that guarantees that $\bm M$ is PSD and shows that this new set is also compact.

\begin{proposition}\label{pr.p3}
Consider the set of $N\times N$ diagonal matrices $\mathcal{D}_{N}$. Let $\Lambdab_{o} \in \mathcal{D}_{N}$ be a fixed diagonal matrix with non-increasing eigenvalues $[\Lambdab_o]_{ii} \geq [\Lambdab_o]_{jj}$ for $i < j$. Let also $ \Lambdab_{\epsilon}\in\mathcal{D}_{N}$ be a diagonal matrix accounting for errors on the elements of $\Lambdab_{o}$ with bounded elements $[\Lambdab_{\epsilon}]_{ii} \in [a_i,b_i]$.

Then, the set
\begin{equation}
    \begin{split}
    \mathcal{M}_{\bm a,\bm b} 
    := \{\bm M \in \mathcal{S}^N~|~ \bm S &= \bm V ( \Lambdab_{o} + \Lambdab_{\epsilon})\bm V^T, \,\bm V \in \bigO(N),\\
    &\; \Lambdab_{\epsilon}\in\mathcal{D}_{N},\, [\Lambdab_{\epsilon}]_{ii} \in [a_i,b_i]\,~\forall i \in [N] \},
    \end{split}
\end{equation}
is compact for vectors $\bm a$ and $\bm b$ such that $[\bm a]_i = a_i$ and $[\bm b]_i = b_i$.
\end{proposition}

\begin{proof}
The proof follows similarly to that of Proposition~\ref{pr.p2} for the set $\mathcal{M}_\epsilon$.

Consider the map
$$
\phi : \bigO(N) \times \prod\limits_{i=1}^{N} [a_i,b_i] \rightarrow \mathcal{M}_{\bm a,\bm b},
$$
which defines the set $\mathcal{M}_{\bm a,\bm b}$ and observe that $\bigO(N)$ and all intervals $[a_i,b_i]$ are compact.

Further, it can be shown that the map $\phi$ is continuous and surjective, thus the set $\mathcal{M}_{\bm a,\bm b}$ is compact. Alternatively, consider that the Cartesian product of compact sets is compact (under the appropriate topology) leading to a compact set $\mathcal{M}_{\bm a,\bm b}$.
\end{proof}

\begin{figure}[!t]
    \centering
    \begin{subfigure}{.25\textwidth}
    \centering
    \psfrag{Error}[cb]{$\Vert\bm S_k - P_{\mathcal{S}_{{\rm NN}}}(\bm S_k)\Vert_{\rm F}$}
    \psfrag{Iteration [k]}[tt]{\scriptsize Iteration $[k]$}
    \includegraphics[width=\textwidth]{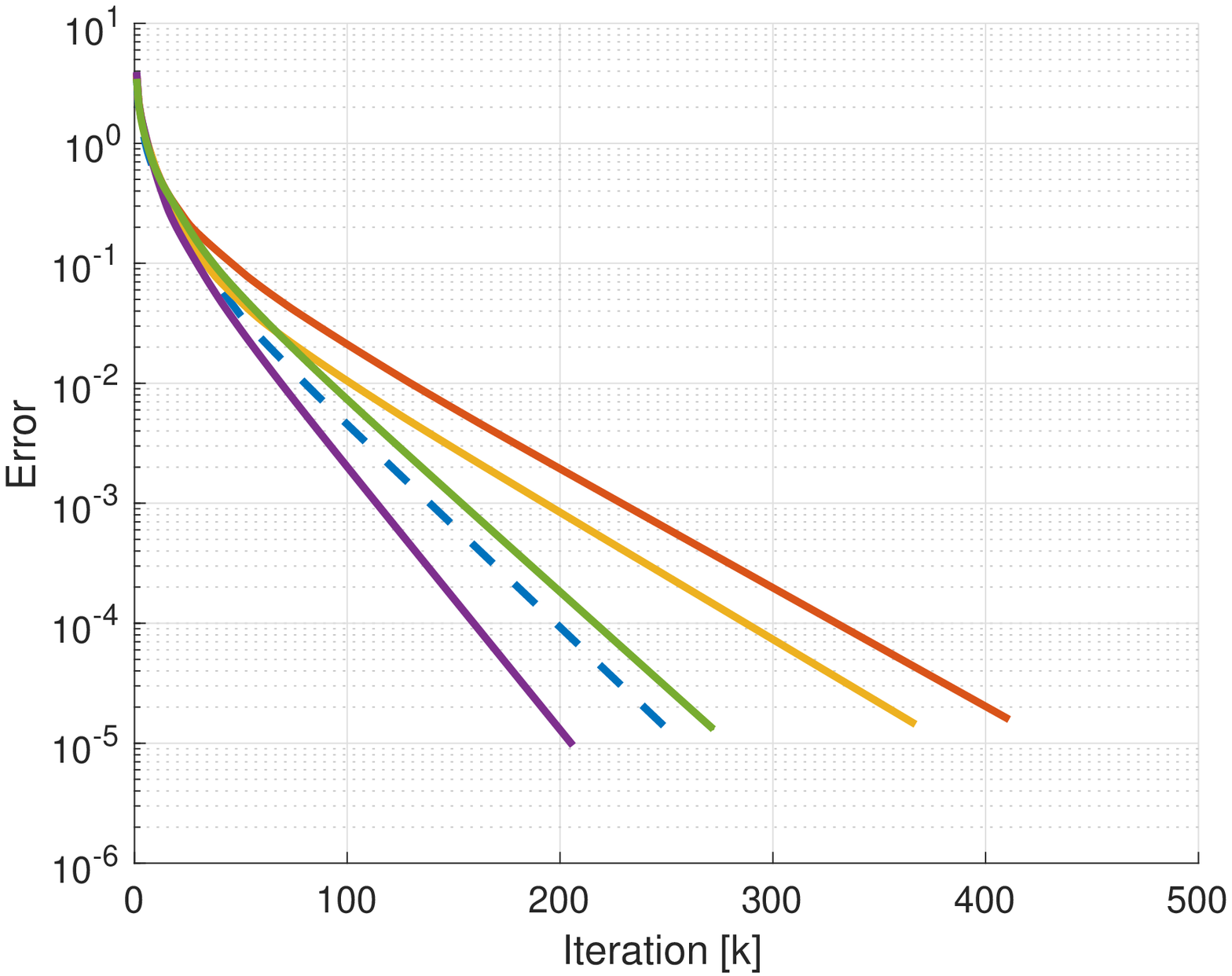}
    \caption{}
    \end{subfigure}%
    \begin{subfigure}{.25\textwidth}
    \centering
    \psfrag{Error}[cb]{$\Vert\bm S_{k} - \bm S_{k+1}\Vert_{\rm F}$}
    \psfrag{Iteration [k]}[tt]{\scriptsize Iteration $[k]$}
    \includegraphics[width=\textwidth]{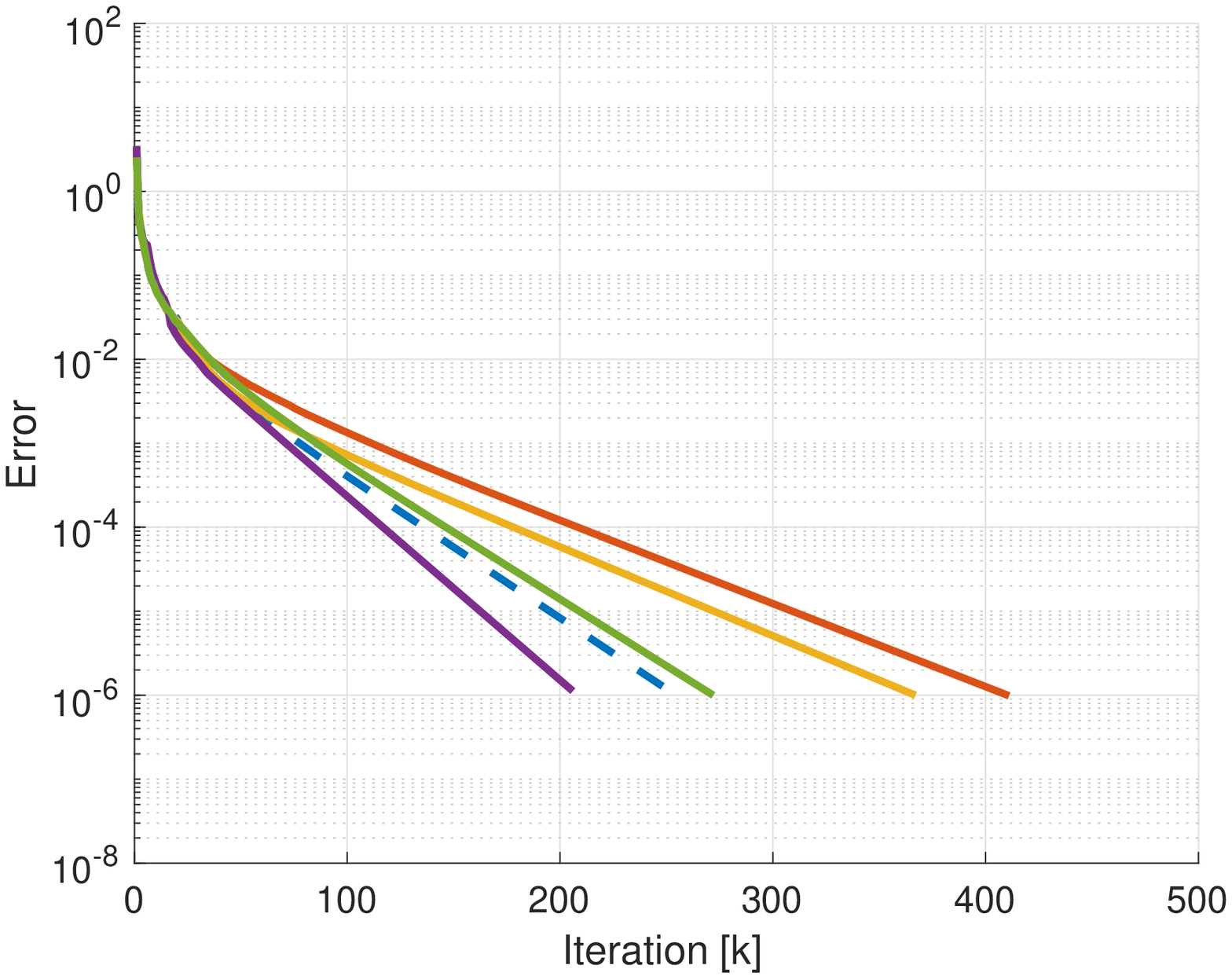}
    \caption{}
    \end{subfigure}
    \caption{Convergence plots for the alternating projections method with $\mathcal{M}$ and $\mathcal{S}_{{\rm NN}}$. (a) Error with respect to the projection, i.e., $\Vert\bm S_k - P_{\mathcal{S}_{{\rm NN}}}(\bm S_k)\Vert_{\rm F}$. (b) Iterate error, i.e., $\Vert\bm S_{k} - \bm S_{k+1}\Vert_{\rm F}$}
    \label{fig.convResults2}
\end{figure}
\subsection{{Proof Theorem~\ref{th.th5}}~(Sketch.)}\label{ap.th5}

We first show that the set $\mathcal{M}_{\epsilon}^m$ is compact. Consider the map
 $$\phi : \bigO(N)\times\mathcal{D}_{m,\epsilon}\times\mathcal{D}_{N-m,\rho}\rightarrow \mathcal{M}_{\epsilon}^m.
 $$
 which defines the set $\mathcal{M}_{\epsilon}^m$ and observe that $\bigO(N)$, $\mathcal{D}_{m,\epsilon}$, and $\mathcal{D}_{N-m,\rho}$ are compact. Then, since the Cartesian product of compact sets is compact (under the appropriate topology), the set $\mathcal{M}_{\epsilon}^m$ is compact.
 
 To prove the optimality of the provided best approximant, we proceed as follows. First, following the construction of~\cite{brockett1991dynamical,chu1992matrix}, it can be shown that for finding the shortest distance from $\Sb$ to $\mathcal{M}_\epsilon^m$, it suffices to find the shortest distance to a particular substructure of $\mathcal{M}_{\epsilon}^m$ defined by the permutation $\bm\sigma$. Then, following the same construction of Theorem~\ref{th.noise}, it can be proven the structure of the optimal $\Lambdab_{\epsilon}^*$.

\subsection{Convergence of Alternating Projections}

In the following, for completeness, we present further results on the convergence of the proposed alternating projections method for the discussed sets in the manuscript. Here, we fix the set $\mathcal{S}$ as the set of nonnegative matrices, $\mathcal{S}_{\rm NN}$. The selection of this feasible set is because this set encompasses the problem of finding stochastic matrices with fixed eigenvalues. 

Similar results as the ones shown for the case of set $\mathcal{M}_{\epsilon}^m$ can be observed. The error with respect to the projection is always monotonic non-increasing, while the iterate error is not necessarily guaranteed to be monotonic. Despite this, in both cases (sets $\mathcal{M}$ and $\mathcal{M}_\epsilon$) the method converges to the desired accuracy of the iterates error.

For these experiments, the employed graph is a $3$-regular graph with $30$ nodes as shown in Fig.~\ref{fig.regGraph}.
\begin{figure}[!t]
    \centering
    \begin{subfigure}{.25\textwidth}
    \centering
    \psfrag{Error}[cb]{$\Vert\bm S_k - P_{\mathcal{S}_{{\rm NN}}}(\bm S_k)\Vert_{\rm F}$}
    \psfrag{Iteration [k]}[tt]{\scriptsize Iteration $[k]$}
    \includegraphics[width=\textwidth]{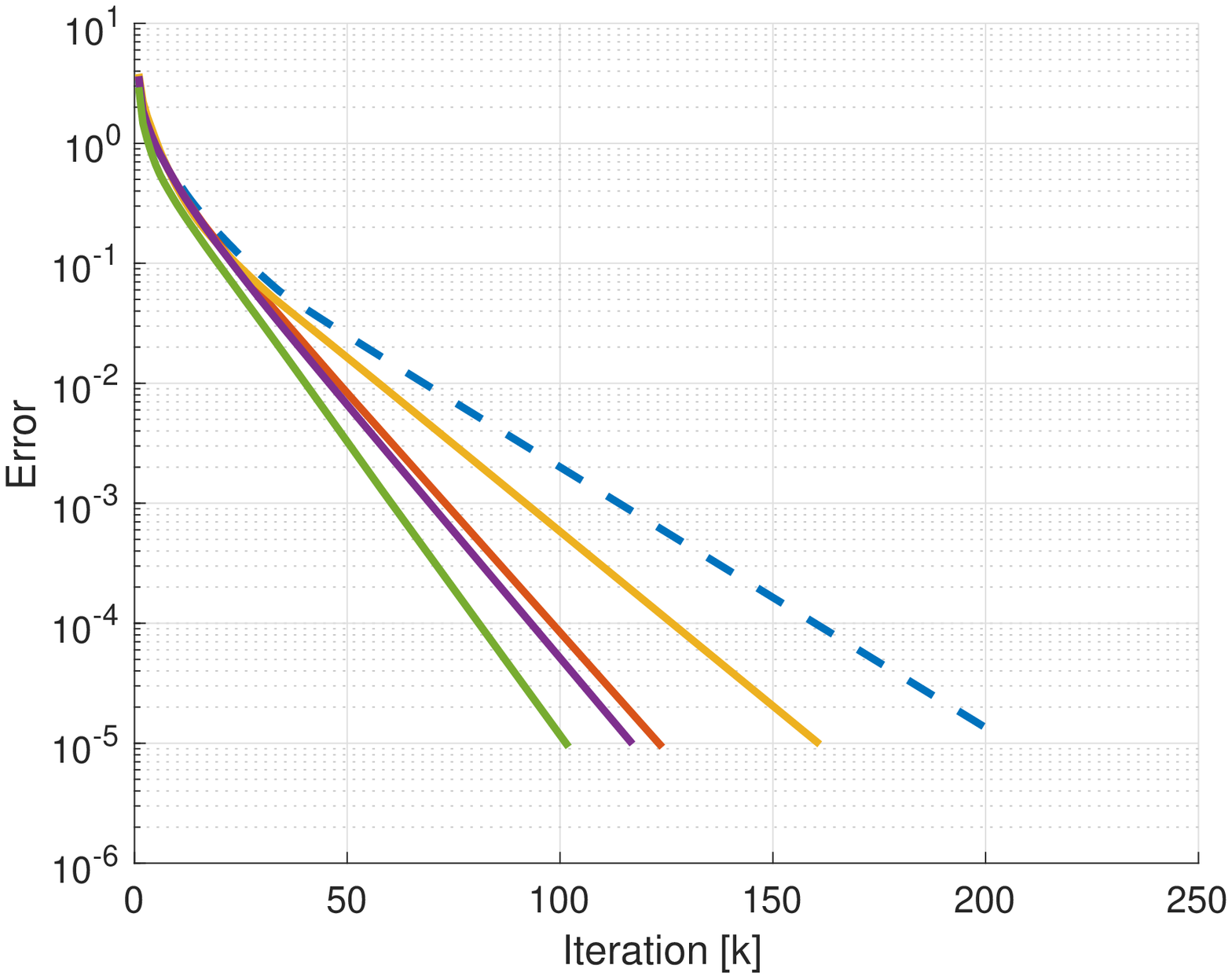}
    \caption{}
    \end{subfigure}%
    \begin{subfigure}{.25\textwidth}
    \centering
    \psfrag{Error}[cb]{$\Vert\bm S_{k} - \bm S_{k+1}\Vert_{\rm F}$}
    \psfrag{Iteration [k]}[tt]{\scriptsize Iteration $[k]$}
    \includegraphics[width=\textwidth]{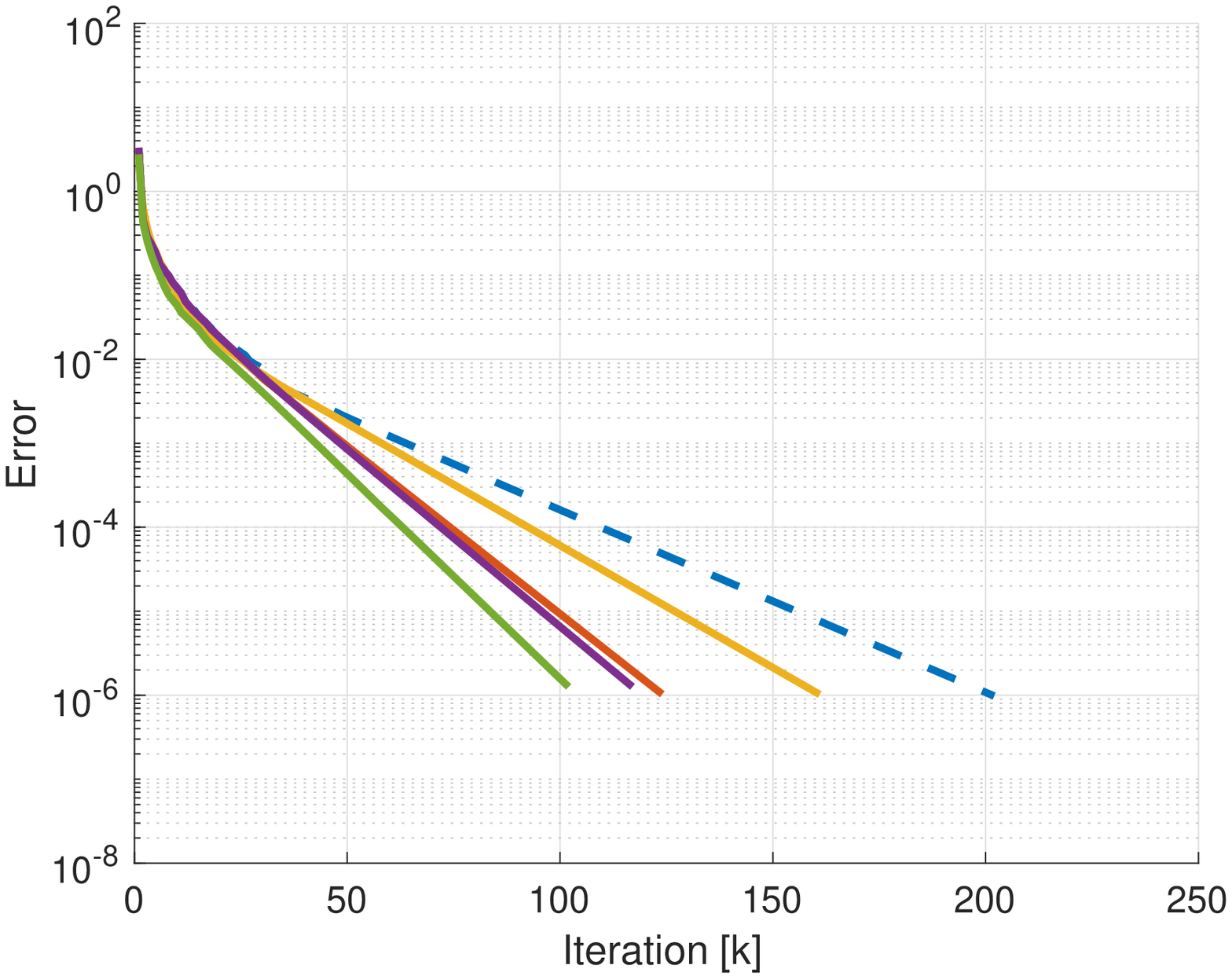}
    \caption{}
    \end{subfigure}
    \caption{Convergence plots for the alternating projections method with $\mathcal{M}_\epsilon$, with $\epsilon = 10^{-1}$, and $\mathcal{S}_{{\rm NN}}$. (a) Error with respect to the projection, i.e., $\Vert\bm S_k - P_{\mathcal{S}_{{\rm NN}}}(\bm S_k)\Vert_{\rm F}$. (b) Iterate error, i.e., $\Vert\bm S_{k} - \bm S_{k+1}\Vert_{\rm F}$}
    \label{fig.convResults3}
\end{figure}

\begin{figure}
    \centering
    \includegraphics[width=0.5\textwidth]{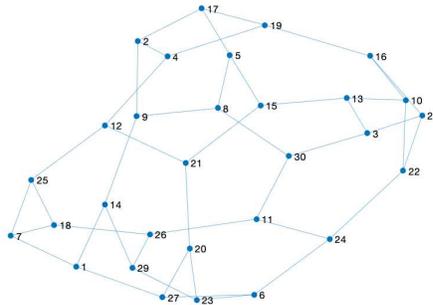}
    \caption{Graph employed for the numerical convergence tests.}
    \label{fig.regGraph}
\end{figure}

\subsection{Response comparison for partially observed network}


Here, a comparison of the predicted output associated with the estimated network structures when compared to the true output of the system for an arbitrary excitation. As seen in Fig.~\ref{fig.estResp}, an approximate $95\%$ of NRMSE is obtained by using the estimated network structures. This s one of the advantages of using the proposed framework. In instances where the problem can be represented by a state-space formulation, finding a graph that satisfies the spectral characteristic of the identified system, within the set of feasible matrix network representations, allows for \emph{consistent} network topology identification.
\begin{figure*}[!t]
    \centering
    \psfrag{Simulated Response Comparison}{\scriptsize Response Comparison}
    \psfrag{simData (y1)}{\tiny Data $[\bm y]_1$}
    \psfrag{simData (y2)}{\tiny Data $[\bm y]_2$}
    \psfrag{simData (y3)}{\tiny Data $[\bm y]_3$}
    \psfrag{simData (y4)}{\tiny Data $[\bm y]_4$}
    \psfrag{simData (y5)}{\tiny Data $[\bm y]_5$}
    \psfrag{simData (y6)}{\tiny Data $[\bm y]_6$}
    \psfrag{simData (y7)}{\tiny Data $[\bm y]_7$}
    \psfrag{Amplitude}{\scriptsize Signal}
    \psfrag{Time (seconds)}[tb]{\scriptsize Time $[t]$}
    \includegraphics[width=\textwidth]{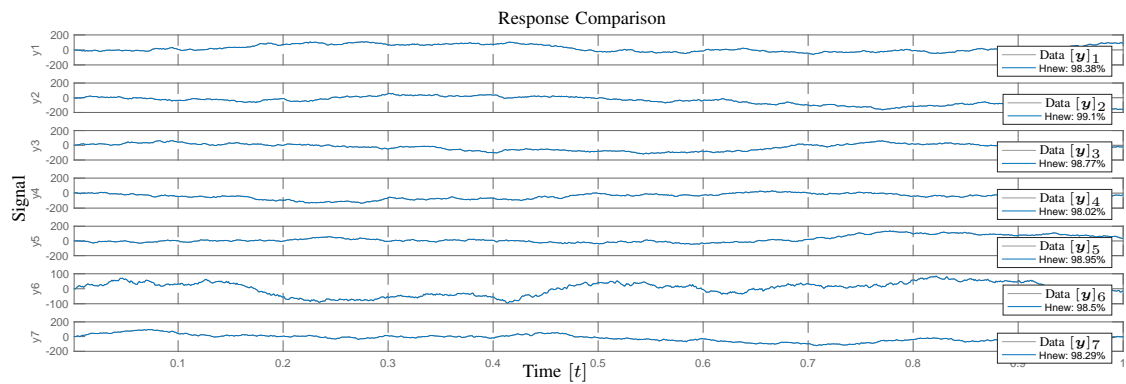}
    \caption{Comparison of the response of the true system to an arbitrary input with the response of the estimated system using the estimated network structures. Here $\H_{\rm new}$ makes reference to the estimated system.}
    \label{fig.estResp}
\end{figure*}
\fi

\end{document}